\documentclass[reqno]{amsart}
\usepackage{amsfonts}
\usepackage{amssymb,latexsym}
\usepackage{amsmath,color}
\usepackage{mathtools}
\usepackage{amsthm}
\usepackage{epsfig}
\usepackage{graphicx}
\usepackage{hyperref}
\usepackage{titletoc}
\usepackage{showkeys}
\usepackage{fullpage}
\usepackage{palatino,mathpazo}
\usepackage{subcaption}

\numberwithin{equation}{section}
\newtheorem{theorem}{Theorem}[section]
\newtheorem{proposition}[theorem]{Proposition}
\newtheorem{lemma}[theorem]{Lemma}

\theoremstyle{definition}

\def\XXint#1#2#3{{\setbox0=\hbox{$#1{#2#3}{\int}$}
     \vcenter{\hbox{$#2#3$}}\kern-.5\wd0}}

\newcommand{\ba}{\begin{array}}
\newcommand{\ea}{\end{array}}
\newcommand{\bed}{\begin{aligned}}
\newcommand{\eed}{\end{aligned}}
\newcommand{\la}{\lambda}

\newcommand{\s}{(-\Delta)^s}
\newcommand{\laps}{(-\Delta)^{s_1}}
\newcommand{\lapr}{(-\Delta)^{s_2}}

\newcommand{\R}{{\mathbb R}}
\newcommand{\e}{\varepsilon}
\newcommand{\p}{{\bf p}}
\newcommand{\ove}{\overline{\varepsilon}}
\newcommand{\ow}{\bar{u}}
\newcommand{\oh}{\bar{v}}
\newcommand{\tw}{\tilde{u}}

\newcommand{\hx}{\hat{\xi}}
\newcommand{\w}{w_{\varepsilon,{\bf p}}}

\newcommand{\tl}{\tilde{L}}
\newcommand{\ov}{\overline}
\newcommand{\dd}{\left(-\frac{1}{\e},\frac{1}{\e}\right)}

\definecolor{orange}{RGB}{255,127,0}

\newcommand{\bx}{\bar{x}}
\newcommand{\by}{\bar{y}}

\begin{document}

\title[fractional GM system in 1D]{Multi-Spike Solutions to the Fractional Gierer-Meinhardt System in a One-Dimensional Domain}

\author[D. Gomez]{Daniel Gomez}
\address{\noindent Daniel Gomez, ~Department of Mathematics, The University of British Columbia, Vancouver, BC Canada V6T 1Z2}
\email{dagubc@math.ubc.ca}

\author[J. Wei]{Jun-cheng Wei}
\address{\noindent Jun-cheng Wei, ~Department of Mathematics, The University of British Columbia, Vancouver, BC Canada V6T 1Z2}
\email{jcwei@math.ubc.ca}

\author[W. Yang]{ Wen Yang}
\address{\noindent Wen ~Yang,~Wuhan Institute of Physics and Mathematics, Chinese Academy of Sciences, P.O. Box 71010, Wuhan 430071, P. R. China}
\email{wyang@wipm.ac.cn}

\begin{abstract}

In this paper we consider the existence and stability of multi-spike solutions to the fractional Gierer-Meinhardt model with periodic boundary conditions. In particular we rigorously prove the existence of symmetric and asymmetric two-spike solutions using a Lyapunov-Schmidt reduction. The linear stability of these two-spike solutions is then rigorously analyzed and found to be determined by the eigenvalues of a certain $2\times 2$ matrix. Our rigorous results are complemented by formal calculations of $N$-spike solutions using the method of matched asymptotic expansions. In addition, we explicitly consider examples of one- and two-spike solutions for which we numerically calculate their relevant existence and stability thresholds.  By considering a one-spike solution we determine that the introduction of fractional diffusion for the activator or inhibitor will respectively destabilize or stabilize a single spike solution with respect to oscillatory instabilities. Furthermore, when considering two-spike solutions we find that the range of parameter values for which asymmetric two-spike solutions exist and for which symmetric two-spike solutions are stable with respect to competition instabilities is expanded with the introduction of fractional inhibitor diffusivity. However our calculations indicate that asymmetric two-spike solutions are always linearly unstable.

\


\end{abstract}

\maketitle

\section{Introduction}

The Gierer-Meinhardt (GM) model is a prototypical activator-inhibitor reaction-diffusion system that has, since its introduction by Gierer and Meinhardt in 1972 \cite{gierer_1972}, been the focus of numerous mathematical studies. In the singularly perturbed limit for which the activator has an asymptotically small diffusivity the GM model is known to exhibit localized solutions in which the activator concentrates at a discrete collection of points and is otherwise exponentially small. The analysis, both rigorous and formal, of the existence, structure, and linear stability of such localized solutions has been the focus of numerous studies over the last two decades (see the book \cite{wei_2014_book}). The GM model in a one-dimensional domain has been particularly well studied using both rigorous PDE methods \cite{wei_1998,wei_2007_existence} as well as formal asymptotic methods \cite{iron_2001,ward_2002_asymmetric}. More recent extensions to the classical one-dimensional GM model have considered the effects of precursors \cite{winter_2009,kolokolnikov_2020}, bulk-membrane-coupling \cite{gomez_2019}, and anomalous diffusion \cite{nec_2012_sub,nec_2012_levi,wei_2019_multi_bump}. It is the latter of these extensions which motivates the following paper which focuses on extending the results obtained in \cite{nec_2012_levi,wei_2019_multi_bump} for the \textit{fractional} one-dimensional GM model. 		

The analysis of localized solutions to the GM model fits more broadly into the study of pattern formation in reaction-diffusion systems. Such reaction-diffusion systems have widespread applicability in the modelling of biological phenomena for which distinct agents diffuse while simultaneously undergoing prescribed reaction kinetics (see classic textbook by Murray \cite{murray_2003}). While these models have typically assumed a normal (or Brownian) diffusion process for which the mean-squared-displacement (MSD) is proportional to the elapsed time, a growing body of literature has considered the alternative of \textit{anomalous diffusion} which may be better suited for biological processes in complex environments \cite{metzler_2004,oliveira_2019,reverey_2015} (see also \S7.1 in \cite{bressloff_2014}). In contrast to normal diffusion, for anomalous diffusion the MSD and time are related by the power law $\text{MSD}\propto (\text{time})^\alpha$ where an exponent satisfying $\alpha>1$ or $\alpha<1$ corresponds to \textit{superdiffusion} or \textit{subdiffusion} respectively. Studies of reaction-diffusion systems with subdiffusion and superdiffusion suggest that anomalous diffusion can have a pronounced impact on pattern formation (see \cite{golovin_2008} and the references therein). In particular studies have shown that both superdiffusion and subdiffusion can reduce the threshold for  Turing instabilities when compared to the same systems with normal diffusion \cite{henry_2005,golovin_2008}. Likewise it has been shown that the Hopf bifurcation threshold for spike solutions to the GM model with normal diffusion for the inhibitor and superdiffusion, mainly with L\'evy flights, for the activator is decreased \cite{nec_2012_levi} whereas it is increased in the case of subdiffusion for the inhibitor and normal diffusion for the activator \cite{nec_2012_sub}.

In this paper we consider the existence and stability of localized multi-spike solutions to the periodic one-dimensional GM model with L\'evy flights for both the activator and the inhibitor. In particular we consider the \textit{fractional} Gierer-Meinhardt system
\begin{equation}
\label{1.fgm}
\begin{cases}
u_t+\e^{2s_1}\laps u+u-\frac{u^2}{v}=0, \quad &\mathrm{for}~x\in(-1,1),\\
\tau v_t+D\lapr v+v-u^2=0, &\mathrm{for}~x\in(-1,1),\\
u(x)=u(x+2),~v(x)=v(x+2),\quad &\mathrm{for}~x\in\R,
\end{cases}
\end{equation}
where $0<\e\ll1$ and the parameters  $0<D<\infty$ and $\tau\geq0$ are independent of $\e.$ We assume the exponents satisfy $1/4<s_1<1$ and $1/2<s_2<1$. The (nonlocal) fractional Laplacian $(-\Delta)^s$ replaces the classical Laplacian as the infinitesimal generator of the underlying L\'evy process for $s<1$ and is defined for all $2$-periodic functions by
\begin{subequations}
\begin{equation}\label{eq:frac_lap_def_1}
(-\Delta)^s\phi(x) \equiv C_s\int_{-\infty}^\infty\frac{\phi(x)-\phi(\bx)}{|x-\bx|^{1+2s}}d\bx = C_s\int_{-1}^1 [\phi(x)-\phi(\bx)]K_s(x-\bx)d\bx,
\end{equation}
where
\begin{equation}\label{eq:frac_lap_def_2}
C_s\equiv \frac{2^{2s}s\Gamma(s+1/2)}{\sqrt{\pi}\Gamma(1-s)},\quad  K_s(z) \equiv \frac{1}{|z|^{1+2s}} + \sum_{j=1}^\infty\biggl(\frac{1}{|z+2j|^{1+2s}}+\frac{1}{|z-2j|^{1+2s}}\biggr),
\end{equation}
\end{subequations}
and for which the second equality in \eqref{eq:frac_lap_def_1} follows from the periodicity of $\phi(x)$. We remark that the system \eqref{1.fgm} closely resembles the system considered in \cite{nec_2012_levi} with the primary difference being that we consider the effects of L\'evy flights for both the activator \textit{and} the inhibitor.

Before outlining the structure of this paper we outline our contributions as follows. Using a Lyapunov-Schmidt type reduction we rigorously prove the existence of symmetric and asymmetric two-spike steady solutions of \eqref{1.fgm} satisfying
\begin{equation}
\label{1.fgms}
\begin{cases}
\e^{2s_1}\laps u+u-\frac{u^2}{v}=0, \quad &\mathrm{for}~x\in(-1,1),\\
D\lapr v+v-u^2=0, &\mathrm{for}~x\in(-1,1),\\
u(x)=u(x+2),~v(x)=v(x+2),\quad &\mathrm{for}~x\in\R.
\end{cases}
\end{equation}
and determine their linear stability by considering the spectrum of certain $2\times 2$ matrices. In addition we use the method of matched asymptotic expansions to formally construct $N$-spike quasi-equilibrium solutions and derive a system of ordinary differential equations governing their slow dynamics. We furthermore illustrate the effects of anomalous diffusion on the stability of one- and two-spike solutions by calculating thresholds for oscillatory and competition instabilities. In particular our results indicate that L\'evy flights for the activator and inhibitor have, respectively, a destabilizing and stabilizing effect on the stability of single spike solutions. On the other hand we demonstrate that the stability of symmetric two-spike solutions with respect to competition instabilities is independent of $s$ and is stabilized when the inhibitor undergoes L\'evy flights. Finally, we show that asymmetric two-spike solutions are always linearly unstable with respect to competition instabilities.

The remainder of this paper is organized as follows. In \S\ref{sec:main-results} we outline the key rigorous results established in this paper pertaining to the existence and stability of two-spike solutions. Then in \S\ref{sec:prelim} we collect preliminary results which are used in the subsequent sections. In \S\ref{sec:formal-results} we use the method of matched asymptotic expansions as well as full numerical simulations to illustrate the effects of fractional diffusion on the structure and stability properties of one- and two-spike solutions. We then provide proofs of the existence and stability results in \S\ref{sec:proof-existence} and \S\ref{sec:proof-stability} respectively. Finally, in \S\ref{sec:conclusion} we make some concluding remarks.

\vspace{0.5cm}
\section{Main results: Existence and Stability}\label{sec:main-results}
In this section we state the main results of this paper, which include the existence of two spike solutions (symmetric and asymmetric) to the steady problem of the fractional Gierer-Meinhardt system and their stability. Instead of studying the system \eqref{1.fgms}, we replace $u(x)$ by $c_\e u(x)$ and $v(x)$ by $c_\e v(x)$, and introduce the scaling $x=\e y$ for the first equation of \eqref{1.fgms}. Then we can write system \eqref{1.fgms} as	
\begin{equation}
\label{1.fg}
\begin{cases}
\s_y u+u-\frac{u^2}{v}=0, \quad &\mathrm{for}~y\in \dd,\\
D\s v+v-{c_\varepsilon}u^2=0, &\mathrm{for}~x\in(-1,1),\\
u(\e y)=u(\e y+2),~v(x)=v(x+2),\quad &\mathrm{for}~x,y\in\R,
\end{cases}
\end{equation}
with
$$c_\e=\left(\e\int_{\R}w^2(y)dy\right)^{-1}$$
and $w$ being the unique solution of
\begin{equation}
\label{1.ground}
\s w+w-w^2=0,\quad w(x)=w(-x).
\end{equation}
From now on, we shall focus on equation \eqref{1.fg} and provide its existence and stability results.

In order to state the main results, we introduce the Green function associate to the steady problem with periodic boundary and make three assupmptions on the Green function that would be used for rigorous proof and stability analysis. For $z\in(-1,1)$, let  $G_D(x,z)$ be the function satisfying
\begin{equation}
\label{3.green}
\begin{cases}
D(-\Delta)^sG_{D}(x,z)+G_{D}(x,z)=\delta(x-z),\quad & \mathrm{for}~x\in(-1,1),\\
G_{D}(x,z)=G_{D}(x+2,z),  &\mathrm{for}~x\in\R,
\end{cases}
\end{equation}
having the Fourier series expansion
\begin{equation*}
G_D(x,z)=\frac{1}{2}\sum_{\ell=-\infty}^\infty\frac{e^{i \ell\pi (x-z)}}{1+D(\ell \pi)^{2s}}=\frac12+\sum_{\ell=1}^\infty\frac{\cos(\ell\pi(x-z))}{1+D(\ell \pi)^{2s}}.
\end{equation*}

Let $-1<p_1^0<p_2^0<1$ be $2$ points in $(-1,1)$ where the spikes concentrate. We introduce several matrices for later use. For $\p=(p_1,p_2)\in(-1,1)^2$ we let $\mathcal{G}_D$ be the $2\times 2$ matrix with entries
\begin{equation}
\label{2.matrix-0}	
(\mathcal{G}_D)_{ij}= G_D(p_i,p_j).	
\end{equation}
Let us denote $\frac{\partial}{\partial p_i}$ as $\nabla_{p_i}$. When $i\neq j$, we can define $\nabla_{p_i}G_D(p_i,p_j)$ in the classical way, while if $i=j$, since $G_D(x,x)$ is a constant due to the periodic boundary condition, we have $\nabla_{p_i}G_D(p_i,p_i)=0$. Next, we define the matrix associated with the first and second derivatives of $\mathcal{G}$ as follows:
\begin{equation}
\label{2.matrix-1}	
\nabla\mathcal{G}_D(\p)=(\nabla_{p_i} G_D(p_i,p_j)),\quad
\nabla^2\mathcal{G}_D(\p)=(\nabla_{p_i}\nabla_{p_j}G_D(p_i,p_j)).
\end{equation}

We make the following two assumptions.
\begin{enumerate}
	\item [(H1)] There exists a solution $(\hat\xi_1^0,\hat\xi_2^0)$
	of the following equation
	\begin{equation}
	\label{2.limit0}
	\sum_{j=1}^2G_D(p_i^0,p_j^0)(\hat\xi_j^0)^2=\hat\xi_i^0,\quad i=1,2.
	\end{equation}
\end{enumerate}	
\begin{enumerate}
	\item [(H2)] $\frac12\notin\lambda(\mathcal{B})$, where $\lambda(\mathcal{B})$ is the set of eigenvalues of the $2\times 2$ matrix $\mathcal{B}$ with entries
	\begin{equation}
	\label{2.matrixB}
	(\mathcal{B})_{ij}=G_D(p_i^0,p_j^0)\hat\xi^0_j
	\end{equation}
\end{enumerate}

By the assumption $(H2)$ and the implicit function theorem, for $\p=(p_1,p_2)$ near $\p^0=(p_1^0,p_2^0)$, there exists a unique solution $\hat\xi(\p)=(\hat\xi_1(\p),\hat\xi_2(\p))$ for the following equation
\begin{equation}
\label{2.limit}
\sum_{j=1}^2 G_D(p_i,p_j)\hat\xi_j^2=\hat\xi_i,\quad i=1,2.
\end{equation}

We define the following vector field:
$$F(\p):=(F_1(\p),F_2(\p)),$$
where
\begin{equation}
\label{2.f}
F_i(\p)=\sum_{j=1}^2\nabla_{p_i}G_D(p_i,p_j)\hat\xi_j^2
=\sum_{j\neq i}\nabla_{p_i}G_D(p_i,p_j)\hat\xi_j^2,~\ i=1,2.
\end{equation}
Set
\begin{equation}
\label{2.m}
\mathcal{M}(\p)=\hat\xi_i^{-1}\nabla_{p_j}F_i(\p).
\end{equation}
The final assumption concerns the vector field $F(\p).$
\begin{enumerate}
	\item [(H3)] We assume that at $\p^0=(p_1^0,p_2^0)$:
	\begin{equation}
	\label{2.reduce}
	F(\p^0)=0\quad \mathrm{and}\quad \mathrm{rank}(\mathcal{M}(\p^0))=1.
	\end{equation}
\end{enumerate}

Next, let us calculate $\mathcal{M}(\p^0)$. Particularly, we shall show that it admits a zero eigenvalue. To compute the matrix $\mathcal{M}(\p^0)$, we have to derive the derivatives of $\hat\xi$. It is easy to see that $\hat\xi(\p)$ is $C^1$ in $\p$ and from \eqref{2.limit} we can calculate:
\begin{equation}
\label{2ij}
\begin{aligned}
\nabla_{p_j}\hat\xi_i=~&2\sum_{l=1}^2G_D(p_i,p_l)\hat\xi_l\nabla_{p_j}\hat\xi_l+\sum_{l=1}^2\frac{\partial }{\partial p_j}G_D(p_i,p_l)\hat\xi_l^2\\
=~&\begin{cases}
2\sum\limits_{l=1}^2G_D(p_i,p_l)\hat\xi_l\nabla_{p_j}\hat\xi_l+\nabla_{p_j}G_D(p_i,p_j)\hat\xi_j^2,\quad &\mathrm{if}~i\neq j,\\
2\sum\limits_{l=1}^2G_D(p_i,p_l)\hat\xi_l\nabla_{p_j}\hat\xi_l+\sum\limits_{l=1}^2\frac{\partial }{\partial p_j}G_D(p_i,p_l)\hat\xi_l^2,\quad &\mathrm{if}~i=j,
\end{cases}
\end{aligned}
\end{equation}
where we used $\partial_{p_i}G_D(p_i,p_i)=0.$ Therefore, if we denote the matrix
\begin{equation}
\label{2.x}
\nabla\xi=(\nabla_{p_j}\hat\xi_i),
\end{equation}
we have
\begin{equation}
\label{2.x-2}	
\nabla\xi(\p)=(I-2\mathcal{G}_D\mathcal{H})^{-1}(\nabla\mathcal{G}_D)^T\mathcal{H}^2+O(\sum_{j=1}^2|F_j(\p)|),
\end{equation}
where a superscript $T$ denotes the transpose and where $\mathcal{H}$ is given by
\begin{equation}
\label{2.h}
\mathcal{H}(\p)=\left(\hat\xi_i(\p)\delta_{ij}\right).
\end{equation}

Let
\begin{equation}
\label{2.q}
\mathcal{Q}=(q_{ij})=(\nabla_{p_i}\nabla_{p_j}G_D(p_1,p_2)\sum_{l\neq i}\frac{\hat\xi_l^2}{\hat\xi_i^2}\delta_{ij}).
\end{equation}
We can compute $\mathcal{M}(\p^0)$ by using \eqref{2ij},
\begin{equation}
\label{2.m0}
\begin{aligned}
\mathcal{M}(\p^0)=\mathcal{H}^{-1}(\nabla^2\mathcal{G}_D+\mathcal{Q})\mathcal{H}^2+2\mathcal{H}^{-1}\nabla\mathcal{G}_D\mathcal{H}(I-2\mathcal{G}_D\mathcal{H})^{-1}
(\nabla\mathcal{G}_D)^T\mathcal{H}^2,
\end{aligned}
\end{equation}
where $A^T$ means the transpose of $A$.	To simplify our notation, we introduce the following matrices:
\begin{equation}
\label{2.p}
\mathcal{P}=(I-2\mathcal{G}_D\mathcal{H})^{-1}.
\end{equation}	
Using \eqref{2.reduce}, we can further simplify the matrix $\mathcal{M}(\p^0)$ as the following
\begin{equation}
\label{2.ms}	
\mathcal{M}(\p^0)=\left(\begin{matrix}
(\hat\xi_1^0)^{-1}\nabla_{p_1}\nabla_{p_1}G_D(p_1,p_2)
(\hat\xi_2^0)^2 & (\hat\xi_1^0)^{-1}\nabla_{p_2}\nabla_{p_1}G_D(p_1,p_2)
(\hat\xi_2^0)^2\\
\\	
(\hat\xi_2^0)^{-1}\nabla_{p_1}\nabla_{p_2}G_D(p_2,p_1)
(\hat\xi_1^0)^2 & (\hat\xi_2^0)^{-1}\nabla_{p_2}\nabla_{p_2}G_D(p_2,p_1)
(\hat\xi_1^0)^2	
\end{matrix}\right).
\end{equation}
It is easy to see that the summation of both rows is zero, thus $\mathcal{M}(\p^0)$ is singular and admits a zero eigenvalue. While the left non-zero eigenvalue can be represented as follows
\begin{equation}
\label{2.me}	
\lambda_{\mathcal{M}(\p^0)}=(\hat\xi_1^0)^{-1}\nabla_{p_1}\nabla_{p_1}G_D(p_1,p_2)(\hat\xi_2^0)^2+(\hat\xi_2^0)^{-1}\nabla_{p_2}\nabla_{p_2}G_D(p_2,p_1)(\hat\xi_1^0)^2.	
\end{equation}

Our first result is the following:

\begin{theorem}
	\label{th1.exist}	
	Assume that (H1) and (H3) are satisfied. Then for $\e\ll1$ problem \eqref{1.fgm} has a 2-spike solution which concentrates at $p_1^\e,p_2^\e$. In addition,
	\begin{equation*}
	u_\e\sim c_\e\sum_{i=1}^2\hat\xi_i^0w\left(\frac{x-p_i^\e}{\e}\right)
	\quad \mathrm{and}\quad v_\e(p_i^\e)\sim c_\e\sum \hat\xi_i^0,~i=1,2,
	\end{equation*}	
	and $(p_1^\e,p_2^\e)\to(-\frac12,\frac12)$ as $\e\to0.$	
\end{theorem}

\noindent{\bf Remark:} In Theorem \ref{th1.exist} the spike height may be the same or different yielding, respectively, \textit{symmetric} and \textit{asymmetric} two-spike solutions. In both cases the spike locations must satisfy $\nabla_{p_1^\varepsilon}G_D(p_1^\varepsilon,p_2^\varepsilon)=0$ and by numerically evaluating the Green's function this implies that $|p_1^\varepsilon-p_2^\varepsilon|=1$. As described in more detail in \S\ref{subsec:example-2} the limiting system \eqref{2.limit} can then be solved explicitly as
	\begin{gather*}
	\xi_1^0=\xi_2^0=\frac{1}{G_D(0,0)+G_D(1,0)}, \quad\text{and}\quad \xi_1^0 = \frac{z_1}{G_D(0,0)},\quad \xi_2^0 = \frac{z_2}{G_D(0,0)},
	\end{gather*}
	for the symmetric and asymmetric cases respectively and where $z_1$ and $z_2$ are defined in terms of $\theta=G_D(1,0)/G_D(0,0)$ in \eqref{eq:z1-z2}.

{Finally, we study the stability of the $2$-spike solution constructed in Theorem \ref{th1.exist}.}

\begin{theorem}
\label{th1.stability}
Assume that $\e\ll1$ and let $(u_\e,v_\e)$ be the solutions constructed in Theorem \ref{th1.exist} and $\mathcal{B}$ be defined in \eqref{2.matrixB}.
\begin{enumerate}
	\item If $\min_{\sigma\in\lambda(\mathcal{B})}\sigma>\frac12$, then there exists $\tau_0$ such that $(u_\e,v_\e)$ is linearly stable for $0\leq\tau<\tau_0$.
	
	\item If $\min_{\sigma\in\lambda(\mathcal{B})}\sigma<\frac12$, then there exists $\tau_0$ such that $(u_\e,v_\e)$ is linearly unstable stable for $0\leq\tau<\tau_0$.
\end{enumerate}
\end{theorem}

\noindent {\bf Remark}: We shall prove Theorem \ref{th1.stability} in section 6. Generally we have to study both large and small eigenvalue problem for the steady state. We shall see that the matrix associated with the small eigenvalues is degenerate: one eigenvalue is zero due to the translational invariance of the spike profiles. On the other hand, the other small eigenvalue is always stable. The stability of the $2$-spike solution therefore depends only on by the matrix $\mathcal{B}$, which naturally appears in the study of large eigenvalue problem.

\noindent {\bf Remark}: To simplify the presentation, in the proof of Theorems \ref{th1.exist} and \ref{th1.stability} we shall only give the details for the case $s_1=s_2=s$. The arguments can be also applied for more general cases where $s_1\in(\tfrac{1}{4},1)$ and $s_2\in(\frac12,1).$

\vspace{0.5cm}
\section{Preliminaries}\label{sec:prelim}

In this section we collect several key preliminary results needed for the existence and stability proofs in \S\ref{sec:proof-existence} and \S\ref{sec:proof-stability} as well as for the formal calculations in \S\ref{sec:formal-results}.

Letting $w$ be the ground state solution satisfying
\begin{equation}
\label{eq:core-problem}
\begin{cases}
(-\Delta)^sw+w-w^2=0,\quad &\mathrm{in}\quad  \R,\\ w(x)\to0~&\mathrm{as}\quad |x|\to \infty,
\end{cases}
\end{equation}
we have the following result \cite{frank_2013_uniqueness} (also see Proposition 4.1 in \cite{wei_2019_multi_bump}  and the references therein)

\begin{proposition}
	\label{pr3.1}
	Equation \eqref{eq:core-problem} admits a positive, radially symmetric solution satisfying the following properties:
	\begin{enumerate}
		\item [(a)] There exists a positive constant {$\mathfrak{b}_s$} depending only on $s$ such that
		$$
		w(x)=\frac{\mathfrak{b}_s}{|x|^{1+2s}}(1+o(1))\quad {\text{as}\quad |x|\rightarrow\infty.}
		$$
		Moreover $w'(x)<0$ for $x>0$ and
		$$w'(x)=-\frac{(1+2s)\mathfrak{b}_s}{x^{2+2s}}(1+o(1))\quad\text{as }\quad x\to\infty.$$ 	
		
		\item [(b)] Let {$L_0=(-\Delta)^s+1-2w$ be the linearized operator. Then} we have
		\begin{equation*}
		\mathrm{Ker}(L_0)=\mathrm{span}\left\{\frac{\partial w}{\partial x}\right\}.
		\end{equation*}

        \item [(c)] Considering the following eigenvalue problem
        $$(-\Delta)^s\phi+\phi-2w\phi+\alpha\phi=0.$$
        There is an unique positive eigenvalue  $\alpha>0.$
	\end{enumerate}
\end{proposition}

Next we consider the stability of a system of nonlocal eigenvalue problems (NLEPs). We first establish the following result which we prove in Appendix \ref{app:nonlocal}.

\begin{theorem}
	\label{th3.stability}	
	Consider the following nonlocal eigenvalue problem
	\begin{equation}
	\label{3.2}
	(-\Delta)^s\phi+\phi-2w\phi+\gamma\frac{\int_{\R}w\phi dx}{\int_{\R}w^2dx}w^2+\alpha\phi=0.
	\end{equation}	
	\begin{enumerate}
		\item [(1)] If $\gamma<1,$ then there is a eigenvalue $\alpha$ to \eqref{3.2} such that $\Re(\alpha)>0.$
		
		\item [(2)] If $\gamma>1$ and $s>\frac14$, then for any nonzero eigenvalue $\alpha$ of \eqref{3.2}, we have
		$$\Re(\alpha)\leq-c_0<0.$$
		
		\item [(3)] If $\gamma\neq1$ and $\alpha=0$, then $\phi=c_0\partial_xw$ for some constant $c_0$.
	\end{enumerate}
\end{theorem}

In our application to the case when $\tau>0,$ we have to deal with the situation when the coefficient $\gamma$ is a function of $\tau\alpha$. Letting $\gamma=\gamma(\tau\alpha)$ be a complex function of $\tau\alpha$ let us suppose that
\begin{equation}
\label{3.3}
\gamma(0)\in\mathbb{R},\quad |\gamma(\tau\alpha)|\leq C~\ \mathrm{for}~\ \alpha_R\geq0,~\tau\geq0,
\end{equation}
where $C$ is a generic constant independent of $\tau,\alpha$. Then we have the following result. 

\begin{theorem}
	\label{th3.2}
	Consider the following nonlocal eigenvalue problem
	\begin{equation}
	\label{3.4}
	(-\Delta)^s\phi+\phi-2w\phi+\gamma(\tau\alpha)\frac{\int_{\R}w\phi dx}{\int_{\R}w^2dx}w^2+\alpha\phi=0,
	\end{equation}	
	where $\gamma(\tau\alpha)$ satisfies \eqref{3.3}. Then there is a small number $\tau_0>0$ such that for $\tau<\tau_0$,
	\begin{enumerate}
		\item [(1)] if $\gamma(0)<1,$ then there is a positive eigenvalue to \eqref{3.4};
		
		\item [(2)] if $\gamma(0)>1$ and $s>\frac14$, then for any nonzero eigenvalue $\alpha$ of \eqref{3.4}, we have
		$$\Re(\alpha)\leq-c_0<0.$$
	\end{enumerate}	
\end{theorem}

\begin{proof}
	The above Theorem follows from Theorem \ref{th3.stability} by a perturbation argument. To make sure that the perturbation works, we have to show that if $\alpha_R\geq0$ and $0<\tau<1$, then $|\alpha|\leq C$, where $C$ is a generic constant (independent of $\tau$). In fact, multiplying \eqref{3.4} by $\bar\phi$ - the conjugate of $\phi$ - and integrating by parts, we obtain that
	\begin{equation}
	\label{3.5}
	\int_{\R}(|(-\Delta)^{\frac{s}{2}}\phi|^2+|\phi|^2-2w|\phi|^2)dx=-\alpha\int_{\R}|\phi|^2-\gamma(\tau\alpha)\frac{\int_{\R}w\phi dx}{\int_{\R}w^2dx}\int_{\R}w^2\overline{\phi}dx.
	\end{equation}	
	From the imaginary part of \eqref{3.5}, we obtain that
	$$|\alpha_{I}|\leq C_1|\gamma(\tau\alpha)|,$$	
	where $\alpha=\alpha_R+\sqrt{-1}\alpha_I$ and $C_1$ is a positive constant (independent of $\tau$). By assumption \eqref{3.3}, $|\gamma(\tau\alpha)|\leq C$ and so $|\alpha_I|\leq C$. Taking the real part of \eqref{3.5} and we get that
	\begin{equation}
	\label{3.6}
	\mbox{left hand side of}~\eqref{3.5}~\geq C\int_{\R}|\phi|^2dx\quad \mbox{for some}~C\in\R,
	\end{equation}
	then we obtain that $\alpha_R\leq C_2$ where $C_2$ is a positive constant (independent of $\tau>0$). Therefore, $|\alpha|$ is uniformly bounded and hence a perturbation argument gives the desired conclusion.
\end{proof}

We now consider the following system of linear operators
\begin{equation}
\label{3.7}
L\Phi:=(-\Delta)^s\Phi+\Phi-2w\Phi+2\mathcal{B}\left(\int_{\R}w {\Phi} dx\right)\left(\int_{\R}w^2dx\right)^{-1}w^2,
\end{equation}
where $\mathcal{B}$ is given by \eqref{2.matrixB} and {$\Phi:=(\phi_1,\phi_2)^T\in (H^{2s}(\R))^2$.} 
The conjugate operator of $L$ under the scalar product in $L^2(\R)$ is
\begin{equation}
\label{3.9}
L^*\Psi:=(-\Delta)^s\Psi+\Psi-2w\Psi+2\mathcal{B}^T\left(\int_{\R}w^2 {\Psi} dx\right)\left(\int_{\R}w^2dx\right)^{-1}w,
\end{equation}
where {$\Psi:=(\psi_1,\psi_2)^T\in (H^{2s}(\R))^2$.} {We then have the following result.}
\begin{lemma}
	\label{le3.1}
	Assume that $(H2)$ holds. Then
	\begin{equation}\label{eq:L_L*_kernel}
	\mathrm{Ker}(L)=\mathrm{Ker}(L^*)=X_0\oplus X_0,
	\end{equation}
	where $X_0=\mathrm{Span}\{w'(x)\}$.
\end{lemma}

\begin{proof}
{We first prove $\mathrm{Ker}(L)\subset X_0\oplus X_0$}. Suppose $L\Phi=0$. 
By the fact that $\mathcal{G}_D$ is symmetry and $\mathcal{H}(\p)$ is a diagonal matrix, we could diagonalize $\mathcal{B}$. Let
$$P^{-1}\mathcal{B}P=\mathcal{J},$$
where $P$ is an orthogonal matrix and $\mathcal{J}$ is diagonal form, i.e.,
	\begin{equation*}
	\mathcal{J}=\left(\begin{matrix}
	\sigma_1 & 0\\
	0&\sigma_2
	\end{matrix}\right)
	\end{equation*}
	with suitable real numbers $\sigma_i,~i=1,2.$ Defining $\Phi=P\tilde\Phi$ we have
	\begin{equation}
	\label{3.13}
	(-\Delta)^s\tilde\Phi+\tilde\Phi-2w\tilde\Phi+2\left(\int_{\R}w^2dx\right)^{-1}\left(\int_{\R}w\mathcal{J}\tilde\Phi dx\right)w^2=0.
	\end{equation}
	For $i=1,2$ we look at the $i$-th equation of system \eqref{3.13}:
	\begin{equation}
	\label{3.14}
	(-\Delta)^s\tilde\Phi_i+\tilde\Phi_i-2w\tilde\Phi_i+2\sigma_i\left(\int_{\R}w^2dx\right)^{-1}\left(\int_{\R}w\tilde\Phi_i dx\right)w^2=0.
	\end{equation}
	By Theorem \ref{th3.stability}-(3),  equation \eqref{3.14} {implies }(since by condition $(H2)$ we know that $2\sigma_i\neq 1$)  {$\tilde\Phi_i\in X_0$.} 

{We proceed similarly to prove $\mathrm{Ker}(L^*)\subset X_0\oplus X_0$.} Using $\sigma(\mathcal{B})=\sigma(\mathcal{B}^T)$ the $i$-th equation of the diagonalized system is as follows
	\begin{equation}
	\label{3.17}
	(-\Delta)^s\tilde\Psi_i+\tilde\Psi_i-2w\tilde\Psi_i+2\sigma_i\left(\int_{\R}w^2dx\right)^{-1}\left(\int_{\R}w^2\tilde\Psi_i dx\right)w=0.
	\end{equation}
	Multiplying the above equation by $w$ and integrating over the real line, we obtain
	\begin{equation}
	\label{3.18}
	(1-2\sigma_i)\int_{\R}w^2\tilde\Psi_i=0,
	\end{equation}
	which together with the fact that $2\sigma_i\neq1$ implies that
	\begin{equation*}
	\int_{\R}w^2\tilde\Psi_i=0,\quad i=1,2.
	\end{equation*}
	Thus all the nonlocal terms vanish and we have $L_0\tilde\Psi_i=0$ for $i=1,2$, {which in turn implies that $\Psi_i\in X_0$ for $i=1,2.$} On the other hand, it is obvious that $X_0\oplus X_0\subset\mathrm{Ker}(L)$ and $X_0\oplus X_0\subset\mathrm{Ker}(L^*)$. Therefore, we conclude that
\eqref{eq:L_L*_kernel} holds.
\end{proof}

\begin{lemma}
	\label{le3.2}
	The operator $L: (H^{2s}(\R))^2\to (L^2(\R))^2$ is invertible if it is restricted as follows
	$$L:(X_0\oplus  X_0)^\perp\cap (H^{2s}(\R))^2\to (X_0\oplus X_0)^\perp\cap (L^2(\R))^2.$$
	Moreover, $L^{-1}$ is bounded.
\end{lemma}

\begin{proof}
	This follows from the Fredholm Alternatives Theorem and Lemma \ref{le3.1}.
\end{proof}

Finally we study the eigenvalue problem (see \eqref{3.7} for the definition of $L$)
\begin{equation}
\label{3.20}
L\Phi+\alpha\Phi=0,
\end{equation}
for which we have the following lemma.

\begin{lemma}
	\label{le3.3}
	Assume that all the eigenvalues of $\mathcal{B}$ are real. Then we have
	\begin{enumerate}
		\item [(1)] If $2\min\limits_{\sigma\in\sigma(\mathcal{B})}\sigma>1$ then for any nonzero eigenvalue of \eqref{3.20} we must have $\Re(\alpha)\leq-c_0<0.$ 
		
		\item [(2)] If there exists $\sigma\in\sigma(\mathcal{B})$ such that $2\sigma<1$, then there exists a positive eigenvalue of \eqref{3.20}.
	\end{enumerate}	
\end{lemma}

\begin{proof}
{We first prove (1)}. Let $(\Phi,\alpha)$ satisfy \eqref{3.20} {and assume that $2\min\limits_{\sigma\in\sigma(\mathcal{B})}\sigma>1$}. Suppose $\alpha_R\geq 0$ and $\alpha\neq 0$. Similar to Lemma \ref{le3.2} we diagonalize \eqref{3.20}
	\begin{equation}
	\label{3.21}
	(-\Delta)^s\Phi+\Phi-2w\Phi+2(\int_{\R}w^2dx)^{-1}(\int_{\R}w \mathcal{J}\Phi)w^2+\alpha\Phi=0,
	\end{equation}
	and the $i$-th equation of system \eqref{3.21} becomes
	\begin{equation}
	\label{3.22}
	(-\Delta)^s\Phi_i+\Phi_i-2w\Phi_i+2\sigma_i\left(\int_{\R}w^2dx\right)^{-1}\left(\int_{\R}w\Phi_i\right)w^2+\alpha\Phi_i=0.
	\end{equation}	
	The first conclusion follows by Theorem \ref{th3.stability}-(2) and the fact that $2\sigma_i>1$. We conclude that either $\Phi_1=\Phi_2=0$ or $\alpha\leq-c_0<0$.	Since $\Phi$ does not vanish and $\alpha<0$, thus (1) is proved.
	
Next we prove (2) and assume that $2\sigma_i<1$ for some $\sigma_i\in {\sigma(\mathcal{B})}$. Then the equation corresponding to $\sigma_i$ becomes
	$$(-\Delta)^s\Phi_i+\Phi_i-2w\Phi_i+2\sigma_i\left(\int_{\R}w^2\right)^{-1}\left(\int_{\R}w\Phi_idx\right)w^2
+\alpha\Phi_i=0.$$
By Theorem \ref{th3.stability}-(1) we know that there exists an eigenvalue $\alpha_0>0$ and an eigenfunction $\Phi_0$ such that
\begin{equation}
\label{3.23} L_0\Phi_0+2\sigma_i\left(\int_{\R}w^2dx\right)^{-1}\left(\int_{\R}w\Phi_0dx\right)w^2+\alpha_0\Phi_0=0.
\end{equation}
Let us take $\Phi_i=\Phi_0$ and $\Phi_j=0$ for $j\neq i$. Then $(\Phi,\alpha_0)$ satisfies \eqref{3.20} {which establishes (2).}
\end{proof}

\vspace{0.5cm}
\section{Formal Analysis of $N$-Spike Equilibrium Solutions and their Linear Stability}\label{sec:formal-results}

Although the fractional Laplacian $(-\Delta)^s$ is \textit{nonlocal}, the method of matched asymptotic expansions can nevertheless be used to construct leading order asymptotic approximations to equilibrium solutions of \eqref{1.fgm}. Indeed, assuming $-1<p_1<...<p_N<1$ ($N\geq 1$) are well separated in the sense that $p_1+1=\mathcal{O}(1)$, $1-p_N=\mathcal{O}(1)$, and $|p_{i+1}-p_i|=\mathcal{O}(1)$ for all $i=1,...,N-1$ then it is clear from the definition \eqref{eq:frac_lap_def_2} that
\begin{equation*}
K_s(p_i+\varepsilon y - p_j - \varepsilon \by) = \begin{cases}  \mathcal{O}(1), & j\neq i, \\ \frac{1}{\varepsilon^{1+2s}}\frac{1}{|y-\bar{y}|^{1+2s}} + \mathcal{O}(1), & j= i, \end{cases}\qquad y,\bar{y}=\mathcal{O}(1).
\end{equation*}
Moreover for any bounded and periodic function $\phi(x)$ such that $\phi(x)\sim\Phi(y)$ for $x=p_i+\varepsilon y$ and $y=\mathcal{O}(1)$
\begin{equation*}
(-\Delta)^s\phi(x) \sim \varepsilon^{-2s}(-\Delta)^s\Phi + \mathcal{O}(1),\qquad (-\Delta)^s\Phi\equiv C_s\int_{-\infty}^\infty \frac{\Phi(y)-\Phi(\by)}{|y-\by|^{1+2s}}d\by,
\end{equation*}
which effectively separates the \textit{inner region} problems in the method of matched asymptotic expansions. In the remainder of this section we use the method of matched asymptotic expansions to formally construct multi-spike equilibrium solutions to \eqref{1.fgm} and determine their linear stability.

\subsection{Multi-Spike Solutions and their Slow Dynamics}\label{subsec:formal-equilibrium}

With the separation of inner region problems as outlined above, the construction of quasi-equilibrium solutions follows closely that for the classical case when $s_1=s_2=1$ as detailed in \cite{iron_2001}. In particular letting $-1<p_1<...<p_N<1$ be given as above, then we obtain the inner expansions
\begin{equation*}
u\sim \varepsilon^{-1}\bigl(\xi_iw_{s_1}(y) + o(1)\bigr),\quad v\sim \varepsilon^{-1}\bigl(\xi_i + o(1)\bigr),\qquad\text{for}\quad x = p_i + \varepsilon y,\quad y=\mathcal{O}(1)
\end{equation*}
for each $i=1,...,N$ where $w_{s_1}$ satisfies the core problem \eqref{eq:core-problem} with $s=s_1$ and $\xi_i>0$ is an undetermined constant. Therefore for all $-1<x<1$
\begin{subequations}\label{eq:quasi-eq-sol}
\begin{equation}\label{eq:quasi-eq-sol-u}
u(x) \sim \varepsilon^{-1}\sum_{i=1}^{N} \xi_i w_{s_1}(\varepsilon^{-1}|x-p_i|) + o(\varepsilon^{-1}),
\end{equation}
where the corrections due to the algebraic decay of the core solution don't contribute until $\mathcal{O}(\varepsilon^{2s_1})$. Moreover, in the sense of distributions we calculate the limit $u^2 \rightarrow \varepsilon^{-1}\omega_{s_1} \sum_{j=1}^N \xi_j^2 \delta(x-p_j)$ as $\varepsilon\rightarrow 0^+$ from which it follows that for all $x$ such that $|x-p_i|\gg\varepsilon$ for all $i=1,...,N$ the inhibitor is given by
\begin{equation}\label{eq:quasi-eq-sol-v}
v\sim \varepsilon^{-1}\omega_{s_1}\sum_{j=1}^{N} \xi_j^2 G_D(x,p_j) + o(\varepsilon^{-1}),\qquad \omega_{s_1} \equiv \int_0^\infty w_{s_1}(y)^2 dy,
\end{equation}
\end{subequations}
where $G_D(\cdot,\cdot)$ is the Green's function satisfying \eqref{3.green} with $s=s_2$. Since $v\rightarrow \varepsilon^{-1}(\xi_i + o(1))$ as $x\rightarrow p_i$ we obtain the nonlinear algebraic system
\begin{subequations}\label{eq:quasi-eq}
\begin{equation}
\pmb{\xi} - \omega_{s_1}\mathcal{G}_D\pmb{\xi}^2 = 0,
\end{equation}
where
\begin{equation}
\pmb{\xi}=(\xi_1,\cdots,\xi_N)^T,\quad \mathcal{G}_{D} = (G_D(p_i,p_j))_{i,j=1}^N.
\end{equation}
\end{subequations}
When $N=2$ we recover the system \eqref{2.limit0} and when $N=1$ we obtain $\xi_1 = [\omega_{s_1}G_D(p_1,p_1)]^{-1}$.

Given a fixed configuration $-1<p_1<...<p_N<1$ the algebraic system \eqref{eq:quasi-eq} can be solved for the unknown constants $\xi_1,...,\xi_N$ yielding \textit{quasi}-equilibrium solution to \eqref{1.fgm} given by \eqref{eq:quasi-eq-sol}. We emphasize that the resulting solutions is not, for arbitrary spike locations, a stationary solution of \eqref{1.fgm}. Indeed, while the solution \eqref{eq:quasi-eq-sol} is stationary over an $\mathcal{O}(1)$ timescale the spike locations drift slowly over an $\mathcal{O}(\varepsilon^{-2})$ timescale according to the system of differential equations (see Appendix \ref{app:slow-dynamics} for details)
\begin{equation}\label{eq:slow-dynamics}
\frac{dp_i}{dt} = -\varepsilon^2\kappa_{s_1} \xi_i^{-1}\sum_{j\neq i}\xi_j^2\nabla_1 G_D(p_i,p_j),\qquad \kappa_{s_1}\equiv \frac{\int_{-\infty}^\infty w_{s_1}^2 dy\int_{-\infty}^\infty w_{s_1}^3 dy}{3\int_{-\infty}^\infty |dw_{s_1}/dy|^2dy},
\end{equation}
where $\nabla_1$ denotes the derivative with respect to the first argument and we remark that this is to be solved concurrently with the algebraic system \eqref{eq:quasi-eq}. In particular, if $-1<p_1<...<p_N<1$ are chosen so that
\begin{equation}
\sum_{j\neq i} \xi_j^2\nabla_1G_D(p_i,p_j)=0,
\end{equation}
for all $i=1,...,N$, then \eqref{eq:quasi-eq-sol} is an \textit{equilibrium} solution of \eqref{1.fgm}. Theorem \ref{th1.exist} and the proof found in \S\ref{sec:proof-existence} rigorously establish the existence of the equilibrium solution constructed in this section.

\subsection{Linear Stability of Multi-Spike Solutions}\label{subsec:formal-stability}

We now consider the linear stability of the $N$-spike equilibrium solutions constructed above which we denote by $u_e$ and $v_e$. Substituting $u=u_e+e^{\lambda t}\phi$ and $v=v_e+e^{\lambda t}\psi$ where $|\phi|,|\psi|\ll 1$ into \eqref{1.fgm} and linearizing we obtain
\begin{subequations}
\begin{gather}
\varepsilon^{2s_1}(-\Delta)^{s_1}\phi + \phi - 2v_e^{-1}u_e\phi + v_e^{-2}u_e^{2}\psi + \lambda \phi = 0,\qquad -1<x<1, \label{eq:lin-stab-phi} \\
D(-\Delta)^{s_2}\psi + \psi - 2u_e\phi + \tau\lambda\psi = 0,\qquad -1<x<1, \label{eq:lin-stab-psi}
\end{gather}
\end{subequations}
where we assume in addition that both $\phi$ and $\psi$ are $2$-periodic. We focus first on the case where $\lambda=\mathcal{O}(1)$, the so-called \textit{large} eigenvalues, and make a brief comment on the case of \textit{small} eigenvalues for which $\lambda=\mathcal{O}(\varepsilon^2)$ at the end of this section. Proceeding with the method of matched asymptotic expansions as in the previous section we deduce that $\phi \sim \phi_i(y) + o(1)$ when $x=p_i+\varepsilon y$ and $y=\mathcal{O}(1)$ for each $i=1,...,N$. It follows that $\phi\sim \sum_{j=1}^N \phi_j(\varepsilon^{-1}(x-p_j))+o(1)$ for all $-1<x<1$ and furthermore $u_e\phi \rightarrow \sum_{j=1}^N\xi_j\int_{-\infty}^\infty w_{s_1}(y)\phi_j(y)dy \delta(x-p_j)$ as $\varepsilon\rightarrow 0^+$ in the sense of distributions. Substituting this into \eqref{eq:lin-stab-psi} we deduce that
\begin{equation*}
\psi(x) = 2\sum_{j=1}^{N} \xi_j\int_{-\infty}^{\infty} w_{s_1}(y)\phi_j(y)dy G_D^\lambda(x,p_j),
\end{equation*}
where $G_D^\lambda(x,z)$ is the eigenvalue dependent Green's function satisfying
\begin{equation}
D(-\Delta)^{s_2} G_D^\lambda + (1+\tau\lambda)G_D^\lambda = \delta(x-z),\quad -1<x,z<1,
\end{equation}
with periodic boundary conditions. It follows that for $x=p_i+\varepsilon y$ equation \eqref{eq:lin-stab-phi} becomes
\begin{equation*}
L_0 \phi_i + 2 w_{s_1}^2 \sum_{j=1}^{N} \xi_j\int_{-\infty}^{\infty} w_{s_1}(y)\phi_j(y)dy G_D^\lambda(p_i,p_j) + \lambda\phi_i = 0,
\end{equation*}
for each $i=1,...,N$ where $L_0$ is the linear operator of Proposition \ref{pr3.1} with $s=s_1$. This system of equations is conveniently rewritten as the system of NLEPs
\begin{subequations}
\begin{equation}\label{eq:nlep-sys}
L_0\Phi + 2w_{s_1}^2\frac{\int_{-\infty}^\infty w_{s_1}\mathcal{E}^\lambda \Phi dy}{\int_{-\infty}^\infty w_{s_1}^2 dy} + \lambda \Phi = 0,
\end{equation}
where
\begin{equation}\label{eq:nlep-sys-def}
\Phi\equiv\begin{pmatrix}\phi_1(y) \\ \vdots \\ \phi_N(y)\end{pmatrix},\quad \mathcal{E}^\lambda = \begin{pmatrix} \hat{\xi}_1 G_D^\lambda(p_1,p_1) & \cdots & \hat{\xi}_N G_D^\lambda(p_1,p_N) \\ \vdots & \ddots & \vdots \\ \hat{\xi}_1 G_D^\lambda(p_N,p_1) & \cdots & \hat{\xi}_N G_D^\lambda(p_N,p_N)\end{pmatrix},\quad \hat{\xi}_i = \omega_{s_1}\xi_i.
\end{equation}
\end{subequations}
Letting $\pmb{p}_k^\lambda$ and $\chi_k^\lambda$ be the eigenpairs of $\mathcal{E}^\lambda$ satisfying $\mathcal{E}^\lambda \pmb{p}_k^\lambda = \chi_k^\lambda\pmb{p}_k^\lambda$ for each $k=0,1,...,N-1$ we can further diagonalize \eqref{eq:nlep-sys} by setting $\Phi = \Phi_k\pmb{p}_k^\lambda$ to get the decoupled system of NLEPs
\begin{equation}\label{eq:nlep-scalar}
L_0\Phi_k + 2\chi_k^\lambda w_{s_1}^2\frac{\int_{-\infty}^\infty w_{s_1}\Phi_k dy}{\int_{-\infty}^\infty w_{s_1}^2dy} + \lambda\Phi_k = 0.
\end{equation}
An $N$-spike equilibrium solution is linearly stable with respect to the large eigenvalues provided that all eigenvalues of \eqref{eq:nlep-scalar} satisfy $\Re(\lambda)<0$ for all $k=0,...,N-1$. Finally, we remark that the NLEP \eqref{eq:nlep-scalar} can be further reduced to the algebraic
\begin{equation}\label{eq:nlep-algebraic}
\mathcal{A}_k(\lambda)\equiv\frac{1}{\chi_k^\lambda} + \mathcal{F}_{s_1}(\lambda) = 0,\qquad \mathcal{F}_{s_1}(\lambda)\equiv 2\frac{\int_{-\infty}^\infty w_{s_1}(L_0+\lambda)^{-1}w_{s_1}^2dy}{\int_{-\infty}^\infty w_{s_1}^2dy},
\end{equation}
which will in general require the numerical evaluation of $F_{s_1}(\lambda)$.

The stability of a multi-spike equilibrium solution with respect to the small eigenvalues is closely related to the slow dynamics given by \eqref{eq:slow-dynamics}. In particular, whereas the large eigenvalues correspond to amplitude instabilities occurring on an $\mathcal{O}(1)$ timescale, the small eigenvalues are linked to the linear stability of the spike pattern with respect to the slow dynamics \eqref{eq:slow-dynamics} and therefore occur on an $\mathcal{O}(\varepsilon^{-2})$ timescale. In the case of two-spike equilibrium solutions Theorem \ref{th1.stability} rigorously establishes the linear stability with respect to the large eigenvalues. On the other hand, as discussed in \S\ref{sec:proof-stability} two-spike equilibrium solutions are always linearly stable with respect to the small eigenvalues. In the remainder of this section we consider explicitly the asymptotic construction and linear stability of one- and two-spike solutions.

\subsection{Example: Symmetric $N$-Spike Solutions}\label{subsec:example-1}

By appropriately choosing the spike locations we can explicitly calculate an $N$-spike solution that is \textit{symmetric} in the sense that the local profile of each spike is identical. Specifically, letting
\begin{equation}
p_i = -1 + N^{-1}(2i-1),\quad \xi_i=\xi_c\equiv \biggl( \omega_{s_1} \sum_{k=0}^{N-1}G_D(2N^{-1}k,0)\biggr)^{-1},\qquad \text{for all}\quad i=1,...,N,
\end{equation}
it is then straightforward to show that \eqref{eq:quasi-eq} is satisfied and the spike locations are stationary solutions of the the slow-dynamics \eqref{eq:slow-dynamics}. Since the resulting matrix $\mathcal{E}^\lambda$ defined by \eqref{eq:nlep-sys-def} is circulant it's eigenpairs are explicitly given by
\begin{equation}
\pmb{p}_k^\lambda \equiv \biggl(1 , e^{i\frac{2\pi k}{N}} , \cdots , e^{i\frac{2\pi (N-1)k}{N}} \biggr)^T,\quad \chi_k^\lambda = \frac{\sum_{j=0}^{N-1} H_j^\lambda e^{i\frac{2\pi jk}{N}}}{\sum_{j=0}^{N-1} H_j^0},\quad H_k^\lambda \equiv G_D^\lambda(2N^{-1}k,0)
\end{equation}
for each $k=0,...,N-1$.

For the remainder of this example we focus exclusively on the calculation of the Hopf bifurcation threshold for a one-spike solution. In particular, by using the winding number argument in \cite{ward_2003}, we seek conditions under which \eqref{eq:nlep-algebraic} with $k=0$ admits an unstable solution (i.e\@  with $\Re(\lambda)>0$). Letting $C_R = \{i\lambda_I\,|\,-R\leq\lambda_I\leq R\}\cup\{Re^{i\theta}\,|\,-\pi/2\leq\theta\leq\pi/2\}$ traversed counterclockwise and noting that $|\chi_0^\lambda|>0$ for all $\lambda$ with $\Re(\lambda)\geq 0$ whereas $\mathcal{F}_{s_1}(\lambda)$ has a simple pole on the positive half-plane corresponding to the principal eigenvalue of $L_0$ we find that the number $Z$ of unstable solutions to \eqref{eq:nlep-algebraic} can be determined by
\begin{equation*}
\frac{1}{2\pi i}\lim_{R\rightarrow\infty}\oint_{C_R}\frac{d\mathcal{A}_0/d\lambda}{\mathcal{A}_0}d\lambda = Z - 1,
\end{equation*}
Noting that $\chi_0^\lambda\sim\mathcal{O}(\lambda^{\tfrac{1}{2s_2}-1})$ and therefore $\mathcal{A}_0(\lambda)\sim\mathcal{O}(\lambda^{1-\tfrac{1}{2s_2}})$ for $|\lambda|\gg 1$ we deduce that the change in argument of $\mathcal{A}_0$ over the semi-circle part of the contour is $\bigl(1-\tfrac{1}{2s_2}\bigr)\pi$ from which it follows that
\begin{equation*}
Z = \tfrac{3}{2} - \tfrac{1}{4s_2} - \tfrac{1}{\pi}\arg\mathcal{A}(i\lambda_I)\bigr|_{\lambda_I=0}^\infty.
\end{equation*}
We note that $\arg\mathcal{A}(i\lambda_I)\rightarrow \tfrac{1}{2}(1-\tfrac{1}{2s_2})$ as $\lambda_I\rightarrow\infty$ whereas $\mathcal{A}_0(0)=-1$ since  $L_0^{-1}w_{s_1}^2=-w_{s_1}$. Furthermore numerical evidence suggests that $\Re\mathcal{A}_0(i\lambda_I)$ is monotone increasing in $\lambda_I$ and so there exists a unique value $0 < \lambda_I^\star < \infty$ such that $\Re\mathcal{A}_0(i\lambda_I^\star) = 0$. It then follows that either $Z=2$ or $Z=0$ depending on whether $\Im\mathcal{A}_0(i\lambda_I^\star)>0$ or $\Im\mathcal{A}_0(i\lambda_I^\star)<0$ respectively. The Hopf bifurcation threshold can thus be calculated by numerically solving $\mathcal{A}_0(i\lambda_I)=0$ for $\tau=\tau_h(D,s_1,s_2)$ and $\lambda_I=\lambda_h(D,s_1,s_2)$. By first considering the limit $D\rightarrow\infty$ for which $\chi_0^\lambda\rightarrow (1+\tau\lambda)^{-1}$ we calculate the Hopf bifurcation threshold $\tau_h^\infty(s_1)$ and accompanying eigenvalue $\lambda_h^\infty(s_1)$, both of which are independent of $s_2$ and are plotted in Figure \ref{fig:hopf_infty}. In particular we observe that $\tau_h^\infty$ is monotone increasing with $s_1$ and therefore the introduction of L\'evy flights for the activator destabilizes the single spike solution as previously observed in \cite{nec_2012_levi}. This behaviour persists for finite values of $D>0$ but we observe that the Hopf bifurcation threshold is monotone decreasing with $s_2$ and therefore introducing L\'evy flights for the inhibitor stabilizes the single spike solution. This behaviour is illustrated in Figure \ref{fig:hopf_thresh} for which we plot the Hopf bifurcation threshold as a function of $D$ for select values of $s_1$ and $s_2$. We remark in addition that the Hopf bifurcation's dependence on the inhibitor diffusivity $D$ remains qualitative unchanged with the introduction of L\'evy flights: $\tau_h(D,s_1,s_2)$ decreases monotonically with $D$.

To illustrate the above observations, mainly the destabilization (resp. stabilization) of the single-spike solution with decreasing $s_1$ (resp. $s_2$), we numerically solve \eqref{1.fgm} starting with a single spike solution centred at $p_1=0$ with $\varepsilon=0.02$, $D=2$, and $\tau=1.5$ for three distinct pairs of exponents $(s_1,s_2)=(0.8,0.7)$, $(0.8,0.9)$, and $(0.4,0.7)$. See Appendix \ref{app:numerical} for details on the numerical calculation. From the numerically calculated threshold we find $\tau_h(2,0.8,0.7)\approx 2.306$, $\tau_h(2,0.8,0.9)\approx 1.096$, and $\tau_h(2,0.4,0.7)\approx 1.399$ and therefore with $\tau_h=1.5$ we anticipate the single spike solution to be stable for the first exponent set and unstable for the latter two. The plots of $u(0,t)$ in Figure \ref{fig:hopf_examples} support these predictions.

\begin{figure}
	\centering
	\begin{subfigure}{0.33\textwidth}
		\centering
		\includegraphics[scale=0.675]{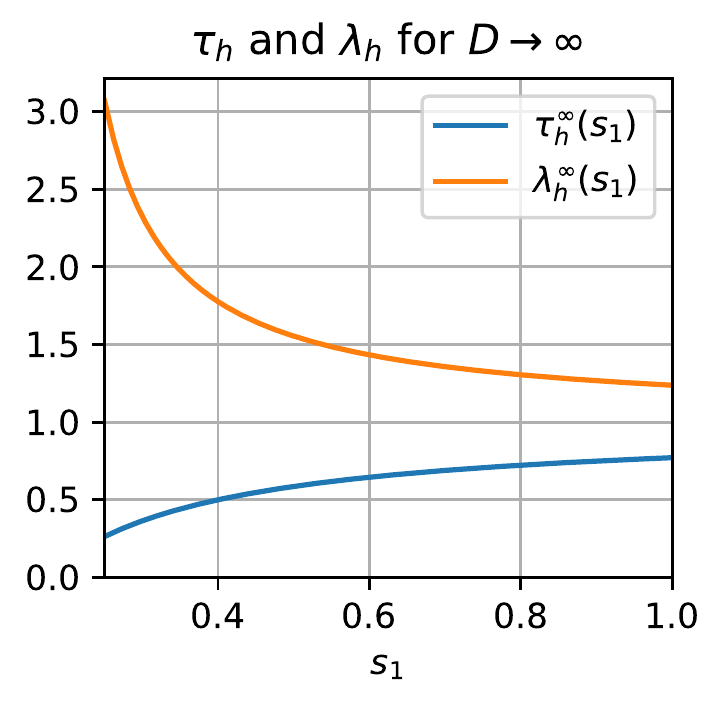}
		\caption{}\label{fig:hopf_infty}
	\end{subfigure}%
	\begin{subfigure}{0.33\textwidth}
		\centering
		\includegraphics[scale=0.675]{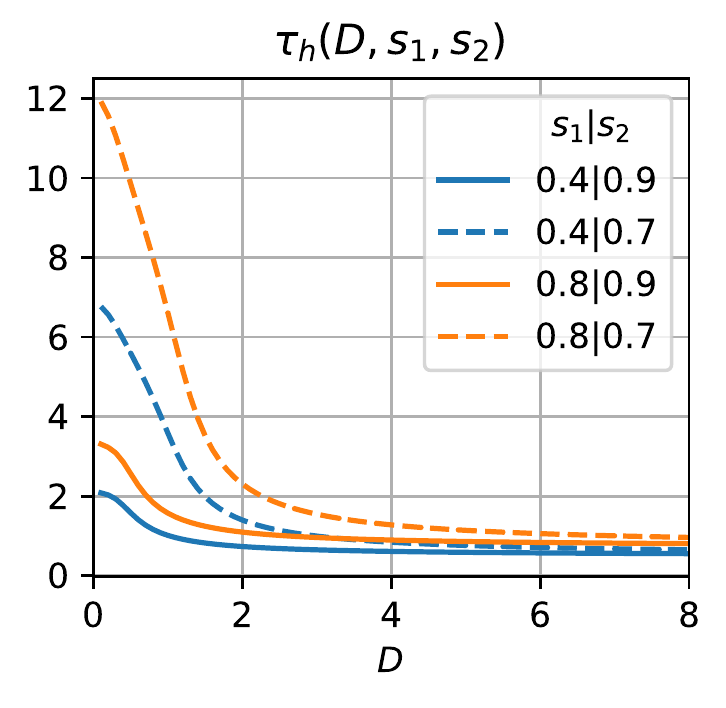}
		\caption{}\label{fig:hopf_thresh}
	\end{subfigure}%
	\begin{subfigure}{0.33\textwidth}
		\centering
		\includegraphics[scale=0.675]{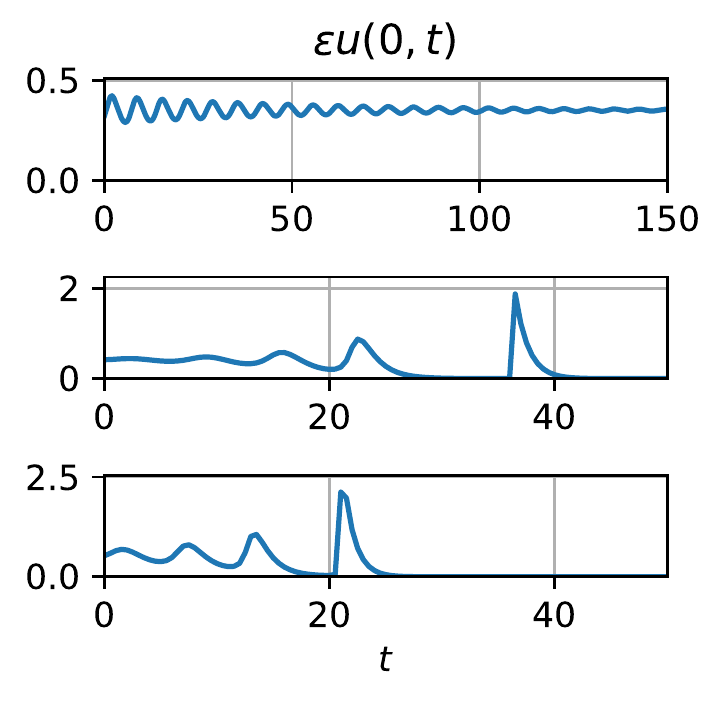}
		\caption{}\label{fig:hopf_examples}
	\end{subfigure}%
	\caption{Hopf bifurcation threshold for a one-spike solution in (A) the shadow limit $D\rightarrow\infty$, and (B) for finite $D>0$ at select values of $s_1=0.4,0.8$ and $s_2=0.7,0.9$. (C) Single spike height obtained from numerical simulations with parameters $\varepsilon=0.02$, $D=2$, and $\tau=1.5$ and exponent sets $(s_1,s_2)=(0.8,0.7)$, $(0.8,0.9)$, and $(0.4,0.7)$ for the top, middle, and bottom plots respectively. }\label{fig:hopf_thresholds}
\end{figure}

\subsection{Example: Symmetric and Asymmetric Two-Spike Solutions}\label{subsec:example-2}

When $s_1=s_2=1$ it has been shown that the one-dimensional Gierer-Meinhardt model may exhibit \textit{asymmetric} solutions consisting of spikes with different heights \cite{ward_2002_asymmetric,wei_2007_existence}. The \textit{gluing} method for constructing such asymmetric $N$-spike solutions relies crucially on the locality of the classical Laplace operator. However, since the fractional Laplace operator $(-\Delta)^s$ is nonlocal for $s<1$ we cannot use this gluing method to construct asymmetric multi-spike solutions and we are therefore restricted to solving the nonlinear algebraic system \eqref{eq:quasi-eq} directly. In this example we restrict our attention to the case of $N=2$ for which a complete characterization of all two-spike solutions can be obtained directly from the algebraic system \eqref{eq:quasi-eq}.

Assuming without loss of generality that $-1<p_1<p_2<1$ we first calculate from \eqref{eq:slow-dynamics} that
\begin{equation*}
\frac{d(p_2-p_1)}{dt} = - \varepsilon^2 \kappa_s \frac{\xi_1^3+\xi_2^3}{\xi_1\xi_2}G_D'(|p_2-p_1|,0),
\end{equation*}
where $G_D'(z,0)=dG_D(z,0)/dz$. By numerically evaluating $G_D(z,0)$ (see Appendix \ref{app:greens-func}) we observe that it is monotone decreasing for $0<z<1$, attains its global minimum at $z=1$, and is monotone increasing for $1<z<2$. Any stationary solution of \eqref{eq:slow-dynamics} must therefore satisfy $p_2-p_1=1$ and furthermore any such solution is linearly stable with respect to the slow-dynamics with the exception of having a neutral eigenvalue corresponding to translational invariance. Defining
\begin{equation*}
z_1 \equiv \omega_{s_1} G_D(0,0)\xi_1,\qquad z_2 \equiv \omega_{s_1} G_D(0,0) \xi_2,\qquad \theta\equiv G_D(1,0)/G_D(0,0),
\end{equation*}
the algebraic system \eqref{eq:quasi-eq} can be rewritten as
\begin{equation}\label{eq:z-system}
z_1 - z_1^2 - \theta z_2^2 = 0,\qquad z_2 - \theta z_1^2 - z_2^2 = 0.
\end{equation}
This system always admits the symmetric solution for which $z_1=z_2=z_c$ where $z_c=(1+\theta)^{-1}$ recovering the result from the previous example for $N=2$. One the other hand, assuming $z_1\neq z_2$ we may subtract the first equation from the second to obtain $z_2 = (1-\theta)^{-1} - z_1$. Substituting this expression for $z_2$ back into the first equation in \eqref{eq:z-system} yields a quadratic in $z_1$ which is readily solved to obtain
\begin{equation}\label{eq:z1-z2}
z_1 = \frac{1/2}{1-\theta}\biggl(1 + \sqrt{\frac{1-3\theta}{1+\theta}}\biggr),\quad z_2 = \frac{1/2}{1-\theta}\biggl(1 - \sqrt{\frac{1-3\theta}{1+\theta}}\biggr).
\end{equation}
We immediately deduce that an asymmetric two-spike solution exists if an only if $\theta<1/3$ and we obtain the bifurcation diagram shown in Figure \ref{fig:two-spike-struct}. Interestingly, the structure of two-spike solutions depends only on the ratio $\theta$ depending only on $D$ and the inhibitor exponent $s_2$.


\begin{figure}
	\centering
	\begin{subfigure}{0.33\textwidth}
		\centering
		\includegraphics[scale=0.675]{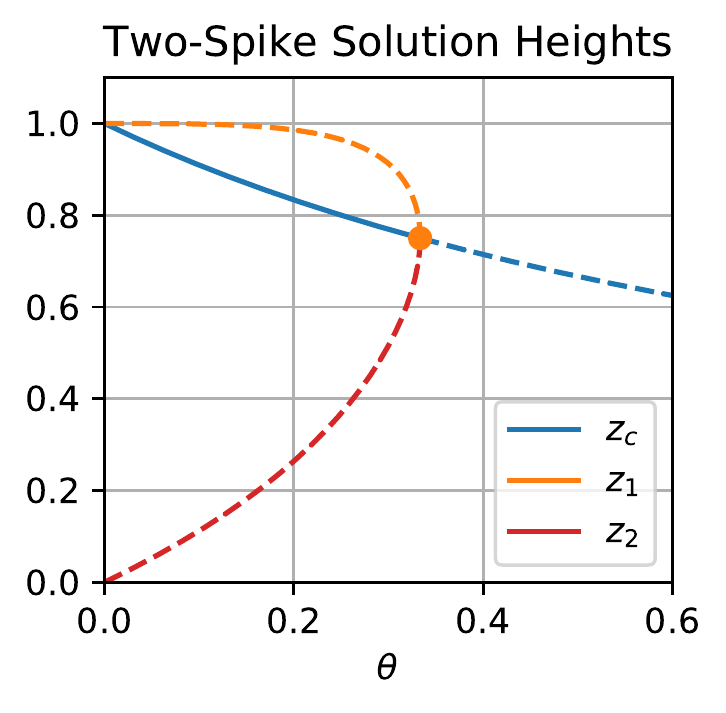}		
		\caption{}\label{fig:two-spike-struct}
	\end{subfigure}%
	\begin{subfigure}{0.33\textwidth}
		\centering
	\includegraphics[scale=0.675]{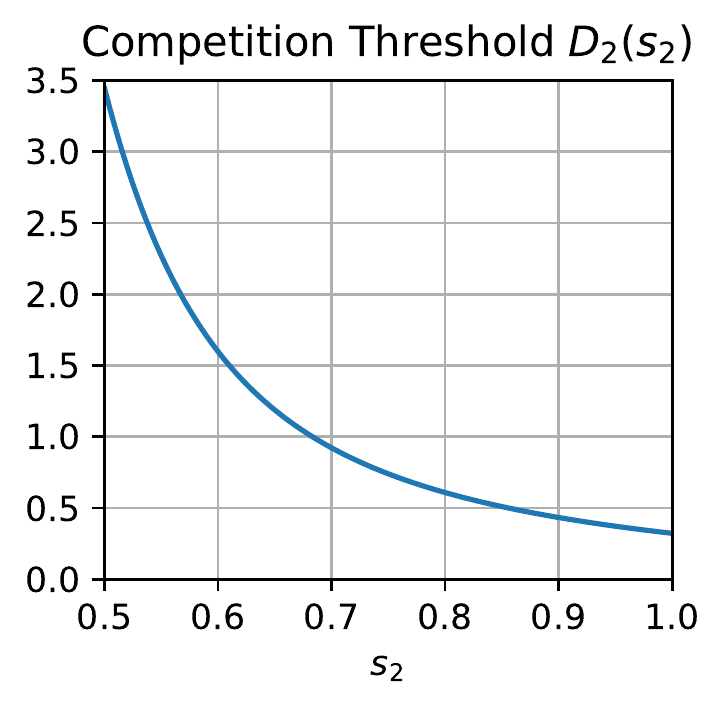}
	\caption{}\label{fig:two-spike-existence}	
	\end{subfigure}%
	\begin{subfigure}{0.33\textwidth}
		\centering
		\includegraphics[scale=0.675]{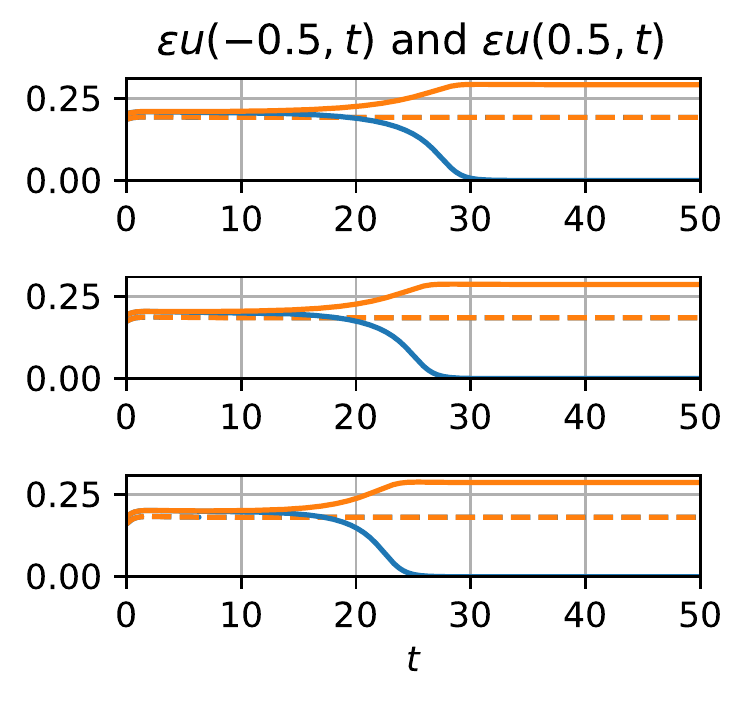}
		\caption{}\label{fig:two-spike-examples}	
	\end{subfigure}%
	\caption{(A) Bifurcation diagram showing the rescaled spike heights $z_i=\omega_{s_1}G_D(0,0)\xi_i$ versus $\theta$. Solid (resp. dashed) lines indicate the resulting two-spike solution is linearly stable (resp. unstable) with respect to competition instabilities. (B) The competition instability threshold for a symmetric two-spike solution. (C) Spike heights at $x=-0.5$ (solid blue) and $x=0.5$ (solid orange) obtained from numerical simulations with initial condition consisting of a symmetric two-spike solution and with parameter values of $s_1=0.8$, $\varepsilon=0.02$, $\tau=0.05$, and $D=1.2D_2(s_2)$  where $s_2=0.0.9$ (top), $0.8$ (middle), and $0.7$ (bottom). The dashed orange line indicates the (common) spike height obtained with the same parameters but with $D=0.8D_2(s_2)$.}\label{fig:two-spike}
\end{figure}

We conclude this section by considering the linear stability of two-spike solutions with respect to competition instabilities, neglecting the possibility of Hopf bifurcations by assuming that $\tau$ is sufficiently small. In view of \eqref{eq:nlep-scalar} and Theorem \ref{th3.stability} it suffices to consider the eigenvalues of
\begin{equation*}
\mathcal{E}_0 = \begin{pmatrix} z_1 & \theta z_2 \\ \theta z_1 & z_2\end{pmatrix}.
\end{equation*}
When $z_1=z_2=z_c$ it is easy to see that $\mathcal{E}_0$ has eigenvectors $\pmb{p}_0 = (1,1)$ and $\pmb{p}_1=(1,-1)$ with corresponding eigenvalues $\chi_0^0=1$ and $\chi_1^0 = (1-\theta)/(1+\theta)$. Since $\chi_0^0>1/2$ by Theorem \ref{th3.stability} the $k=0$ mode is always linearly stable. On the other hand the $k=1$ mode is stable if and only if $\chi_1^0 < 1/2$ and in particular the symmetric solution is linearly stable when $\theta<1/3$ and unstable otherwise. Finally, when $z_1$ and $z_2$ are given by \eqref{eq:z1-z2} it can be shown that the eigenvalues of $\mathcal{E}_0$ are given by
\begin{equation*}
\chi_0^0 = \frac{1/2}{1-\theta}\biggl(1 + \sqrt{\frac{4\theta^2 - 3\theta + 1}{1+\theta}}\biggr),\quad \chi_1^0 = \frac{1/2}{1-\theta}\biggl(1 - \sqrt{\frac{4\theta^2 - 3\theta + 1}{1+\theta}}\biggr),
\end{equation*}
from which we deduce that $\chi_0^0 > 1$ and $\chi_1^0 < \frac{3-\sqrt{3}}{4}<\frac{1}{2}$ for all $0<\theta<1/3$. Therefore by Theorem \ref{th3.stability} the $k=1$ mode is linearly unstable. In Figure \ref{fig:two-spike-struct} we indicate the values of $\theta$ where the two-spike solution is linearly stable (resp. unstable) with respect to competition instabilities by solid (resp. dashed) curves. By numerically solving $\theta=1/3$ for $D$ as a function of $s_2$ we can calculate the competition instability threshold $D=D_2(s_2)$ for the symmetric two-spike solution and this is shown in Figure \ref{fig:two-spike-existence}. In Figure \ref{fig:two-spike-examples} we illustrate the onset of competition instabilities when $s_1=0.8$, $\varepsilon=0.02$, $\tau=0.05$, and for values of $s_2=0.0.9$, $0.8$, and $0.7$ and $D=1.2\times D_2(s_2)$ by performing full numerical simulations of \eqref{1.fgm} (see Appendix \ref{app:numerical} for details). We remark that the accuracy of the leading order approximation to the competition instability calculated above grows increasingly inaccurate as $s_2\rightarrow 0.5$ for a fixed value of $\varepsilon>0$. Indeed, as described in more detail in the derivation of the slow dynamics found in Appendix \ref{app:slow-dynamics}, the first order correction to the quasi-equilibrium solution is $\mathcal{O}(\varepsilon^{2s_2-1})$ and this tends to $\mathcal{O}(1)$ as $s_2\rightarrow 1/2$. When $s_2=1/2$ the Green's function is known to have a logarithmic singularity (see Lemma 2.2. in \cite{wei_2019_multi_bump}) and we anticipate that the method of matched asymptotic expansions will lead to an asymptotic expansion in powers of $\nu=-1/\log\varepsilon$ as is often the case for singularly perturbed reaction-diffusion systems in two-dimensions \cite{kolokolnikov_2009,chen_2011}.


\vspace{0.5cm}
\section{Rigorous proof of the existence results}\label{sec:proof-existence}

In this section we shall prove the existence theorem, i.e., Theorem \ref{th1.exist}. We divide the discussion into three sections. In first subsection, we give an approximate solution. Then we apply the classical Liapunov-Schmidt reduction method to reduce the infinite dimensional problem to a finite dimensional problem in second subsection. In last subsection we solve the finite dimensional problem and thereby prove the Theorem \ref{th1.exist}.

\subsection{Study of the Approximate Solutions}

Let $-1<p_1^0<p_2^0<1$ be $2$ points satisfying the assumptions $(H1)-(H3)$. Let $\hat\xi^0=(\hat\xi_1^0,\hat\xi_2^0)$ be the solution of \eqref{2.limit0} and let $\p^0=(p_1^0,p_2^0)$. We shall construct an approximate solution to \eqref{1.fg} which concentrates near these $2$ points.

Let $-1<p_1<p_2<1$ be such that $\p=(p_1,p_2)\in B_{\e^{2s-1}}{(\p^0)}$. Set
$$r_0=\frac{1}{10}\min\left\{p_1^0+1,~1-p_2^0,~\frac12|p_1^0-p_2^0|\right\}$$
and define a cut-off function $\chi(x)$ such that $\chi(x)=1$ for $|x|<1$ and $\chi(x)=0$ for $|x|>2$. Letting
\begin{equation}
\label{4.1}	
w_i(y)=w\left(y-\frac{p_i}{\e}\right)\chi\left(\frac{\e y-p_i}{r_0}\right),
\end{equation}
{where $w$ is the ground state solution of \eqref{eq:core-problem},} it is then straightforward to check that
\begin{equation}
\label{4.2}	
\s_y w_i(y)+w_i(y)-w_i^2(y)=h.o.t.,
\end{equation}
where $h.o.t.$ {refers to terms} of order $\e^{1+2s}$ in $L^2\dd$. Let $\hat\xi(\p)=(\hat\xi_1,\hat\xi_2)$ be defined as in $(H1)$. Fix any function $u\in H^{2s}\dd$ and define $T[u]$ to be the solution of
\begin{equation}
\label{4.3}
\begin{cases}
D(-\Delta)^sT[u]+T[u]-c_{\e}u^2=0,\quad &x\in(-1,1),\\
T[u](x)=T[u](x+2),\quad &x\in\mathbb{R},
\end{cases}
\end{equation}
where
\begin{equation}
\label{4.4}
c_\e=\left(\e\int_{\R}w^2(y)dy\right)^{-1}.
\end{equation}
\medskip

Letting $\p\in B_{\e^{2s-1}}(\p^0)$ we define $$w_{\e,\p}=\sum_{i=1}^2\hat\xi_iw_i(y)$$ and using \eqref{4.3} we compute
\begin{equation}
\label{4.6}
\begin{aligned}
\tau_i:=T[\w](p_i)=~&\e c_\e\int_{-\frac{1}{\e}}^{\frac{1}{\e}}G_D(p_i,\e y)\w^2(y)dy
=\e c_\e\sum_{j=1}^2\hat\xi_j^2
\int_{-\frac{1}{\e}}^{\frac{1}{\e}}G_D(p_i,\e y)w_j^2(y)dy \\
=~&\e c_{\e}\sum_{j=1}^2\hx_j^2\left(G_D(p_i,p_j)\int_{\R}w^2(y)dy\right)+\mathcal{P}_i=\sum_{j=1}^2G_D(p_i,p_j)\hx_j^2+ \mathcal{P}_i,
\end{aligned}
\end{equation}
where $G_D(x,y)$ is defined in \eqref{3.green} and
$\mathcal{P}_i$ is a number with order $\e^{2s-1}$.
Thus, we have obtained the following system of equations:
\begin{equation}
\label{4.7}
\tau_i=\sum_{j=1}^2G_D(p_i,p_j)\hx_j^2+\mathcal{P}_i.
\end{equation}
According to the assumption $(H1)$-$(H3)$ and the implicit function theorem, we have the above equation has a unique solution
$$\tau_i=\hat\xi_i+\vartheta_i,\quad i=1,2,\quad \vartheta_i=O(\e^{2s-1}).$$
Hence
\begin{equation*}
T[\w](p_i)=\hat\xi_i+O(\e^{2s-1}).
\end{equation*}
Now for $x=p_i+\e z$ we calculate 
\begin{equation}
\label{4.8}
\begin{aligned}
&T[\w](x)-T[\w](p_i)=c_{\e}\int_{-1}^1[G_D(x,\zeta)-G_D(p_i,\zeta)]\w^2\left(\frac{\zeta}{\e}\right)d\zeta\\
&=c_\e\hat\xi_i^2\int_{-1}^{1}[G_D(x,\zeta)-G_D(p_i,\zeta)]w_i^2\left(\frac{\zeta}{\e}\right)d\zeta+c_\e\sum_{j\neq i}\hat\xi_j^2\int_{-1}^1[G_D(x,\zeta)-G_D(p_i,\zeta)]w_j^2\left(\frac{\zeta}{\e}\right)d\zeta \\
&=c_\e\hat\xi_i^2\int_{\R}[G_D(\e y-\e z)-G_D(\e y)]w^2(y)dy+c_\e\sum_{j\neq i}\hat\xi_j^2\int_{-1}^1[G_D(x,\zeta)-G_D(p_i,\zeta)]w_j^2\left(\frac{\zeta}{\e}\right)d\zeta+h.o.t.\\
&=P_i(z)+\e\sum_{j\neq i}\left(\hat\xi_j^2z\nabla_{p_i}G_D(p_i,p_j)+O(\e z^2)\right)+h.o.t.,
\end{aligned}
\end{equation}
where
$$P_i(z)=c_\e\hat\xi_i^2\int_{\R}[G_D(\e y-\e z)-G_D(\e y)]w^2(y)dy~\mbox{is an even function and of order}~\e^{2s-1}.$$

{Next we define}
\begin{equation}
\label{4.9}
S[u]:=(-\Delta)_y^su+u-\frac{u^2}{T[u]},
\end{equation}
{for which we calculate} 
\begin{equation}
\label{4.10}
\begin{aligned}
S[w_{\e,\p}](y)=~&\s_y w_{\e,\p}+w_{\e,\p}-\frac{w_{\e,\p}^2}{T[w_{\e,\p}]}\\
=~&\sum_{j=1}^2\hat\xi_j\chi\left(\frac{\e y-p_j}{r_0}\right)\s w\left(y-\frac{p_j}{\e}\right)+\sum_{j=1}^2\hat\xi_jw_j
-\frac{w_{\e,\p}^2}{T[w_{\e,\p}]}+h.o.t.\\
=~&\left[\sum_{j=1}^2\hat\xi_j{w}_j^2-\frac{(\sum_{j=1}^2\hat\xi_jw_j)^2}{T[w_{\e,\p}]}\right]+h.o.t.\\
=~&E_1+E_2+h.o.t.\quad \mathrm{in}\quad L^2\left(-\frac{1}{\e},\frac{1}{\e}\right),
\end{aligned}
\end{equation}
where
\begin{equation*}
E_1=\sum_{j=1}^2\hat\xi_jw_j^2-\frac{(\sum_{j=1}^2\hat\xi_jw_j)^2}{T[w_{\e,\p}](p_i)},
\quad\mathrm{and}\quad
E_2=\frac{(\sum_{j=1}^2\hat\xi_jw_j)^2}{T[w_{\e,\p}](p_i)}-\frac{(\sum_{j=1}^2\hat\xi_jw_j)^2}{T[w_{\e,\p}](x)}.
\end{equation*}
{Using \eqref{4.7} we calculate}
\begin{equation*}
E_1=\sum_{j=1}^2\hat\xi_jw_j^2-\frac{(\sum_{j=1}^2\hat\xi_jw_j)^2}{T[w_{\e,\p}](p_i)}
=\sum_{j=1}^{2}\left(\hat\xi_j-\frac{\hat\xi_j^{2}}{\hat\xi_i+\vartheta_i}\right)w_j^2
=O(\e^{2s-1})\sum_{j=1}^2\hat\xi_jw_j^2,
\end{equation*}
{and therefore}
\begin{equation}
\label{4.11}
\|E_1\|_{L^2\dd}=O(\e^{2s-1}).
\end{equation}
In addition since $x$ is close to $p_i$ we see that $E_1$ can be decomposed into two parts: one part of order $\e^{2s-1}$ and symmetric in $x-p_i$, and the other part of order $\e$. Next we calculate
\begin{equation}
\label{4.12}
\begin{aligned}
E_2=~&\sum_{j=1}^2\frac{(\hat\xi_jw_j)^2}{(T[w_{\e,\p}](p_i))^2}
(T[w_{\e,\p}](x)-T[w_{\e,\p}](p_i))\left(1+\sum_{n=1}^\infty\left(\frac{T[w_{\e,\p}](p_i)-T[w_{\e,\p}](x)}{T[w_{\e,\p}](p_i)}\right)^n\right)\\
=~&\sum_{j=1}^2\frac{(\hat\xi_jw_j)^2}{(T[w_{\e,\p}](p_i))^2}P_i(z)\left(1+\sum_{n=1}^\infty\left(\frac{P_i(z)}{T[\w](p_i)}\right)^n\right)\\
&+\e\sum_{j=1}^2w_j^2\sum_{l\neq i}\hat\xi_l^2z\nabla_{p_i}G(p_i,p_l)+h.o.t.\\
=~&E_{21}+E_{22}+h.o.t.,
\end{aligned}
\end{equation}
where 
\begin{equation}
\label{4.13}
E_{21}=O(\e^{2s-1})~\mbox{is symmetry in}~x-p_i,~i=1,2,
\quad \mathrm{and}\quad
\|E_{22}\|_{L^2\left(-\frac{1}{\e},\frac{1}{\e}\right)}=O(\e).
\end{equation}
{We have thus established the following lemma.}
\begin{lemma}
	\label{le4.1}
	For $x=p_i+\e z$, $|\e z|<r_0$, we have the decomposition for $S[\w](x)$,
	$$S[\w]=S_{1,1}+S_{1,2},$$
	where
	\begin{equation*}
	S_{1,1}(z)=\e\sum_{j=1}^2w_j^2\sum_{l\neq i}\hat\xi_l^2z\nabla_{p_i}G(p_i,p_l)+h.o.t.,
	\end{equation*}
	and
	\begin{equation*}
	S_{1,2}(z)=\frac{(\hat\xi_iw_i)^2}{(T[\w](p_i))^2}R_i(z)+h.o.t.,
	\end{equation*}
	where $R_i(z)$ is even in $z$ and $\|S_{1,2}\|_{L^2\dd}\leq C\e^{2s-1}$. Furthermore,
	$$S[\w]=h.o.t.\quad \mbox{for}\quad  |x-p_i|\geq r_0,\quad i=1,2.$$
\end{lemma}

\medskip
\subsection{The Liapunov-Schmidt Reduction Method}
In this subsection, we use the Liapunov-Schmidt reduction method to solve the problem
\begin{equation}
\label{5.1}
S[w_{\e,\p}+\phi]=\sum_{j=1}^2c_j\frac{\partial w_j}{\partial y}
\end{equation}
for real constants $c_j$ and a perturbation $\phi\in H^{2s}\dd$ which is small in the corresponding norm. {To proceed} we {first} need to study the linearized operator
$$\tilde{L}_{\e,\p}\phi:=S_{\e}'[\w]\phi=(-\Delta)_y^s\phi+\phi-2\frac{\w}{T[\w]}\phi+\frac{\w^2}{(T[\w])^2}(T'[\w]\phi).$$
For a given function $\phi\in L^2(\Omega)$ we introduce $T'[\w]\phi$ as the unique solution of
\begin{equation}
\label{5.2}
\begin{cases}
D(-\Delta)^s(T'[\w]\phi)+T'[\w]\phi-2c_\e \w\phi=0,\quad &x\in(-1,1),\\
(T'[\w]\phi)(x)=(T'[\w]\phi)(x+2), &x\in \R.
\end{cases}
\end{equation}

{The approximate kernel and co-kernel are respectively defined by}
\begin{gather*}
\mathcal{K}_{\e,\p}:=\mathrm{Span}\left\{\frac{\partial w_j}{\partial y}\Big| ~j=1,2\right\}\subset H^{2s}\dd,\\
\mathcal{C}_{\e,\p}:=\mathrm{Span}\left\{\frac{\partial w_j}{\partial y}\Big|~ j=1,2\right\}\subset L^2\dd.
\end{gather*}
{From the definition of the linear operator $L$ in \eqref{3.7} we recall that by {Lemma \ref{le3.2}} we know that
	\begin{equation*}
	L:~(X_0\oplus X_0)^\perp\cap (H^{2s}(\R))^2\to
	(X_0\oplus X_0)^\perp\cap (L^2(\R))^2
	\end{equation*}
	is invertible with a bounded inverse.} We shall see that {the linear operator $L$} is a limit of the operator $\tl_{\e,\p}$ as $\e\to0$. {First we } introduce the projection $\pi_{\e,\p}^\perp:L^2\dd\to\mathcal{C}_{\e,\p}^\perp$ and study the operator $L_{\e,\p}:=\pi_{\e,\p}^\perp\circ \tl_{\e,\p}$. Letting $\e\to0$ we shall show that $L_{\e,\p}:\mathcal{K}_{\e,\p}^\perp\to\mathcal{C}_{\e,\p}^\perp$ is invertible with a bounded inverse provided $\e$ is small enough. This result is contained in the following proposition.

\begin{proposition}
	\label{pr5.1}
	There exists positive constants $\overline{\e},\overline{\delta},C$ such that for all $\e\in(0,\overline{\e})$, $(p_1,p_2)\in(-1,1)^2$ with $\min(|1+p_1|,~|1-p_2|,~|p_1-p_2|)>\overline{\delta},$
	\begin{equation*}
	\|L_{\e,\p}\phi\|_{L^2\left(-\frac{1}{\e},\frac{1}{\e}\right)}\geq C\|\phi\|_{H^{2s}\left(-\frac{1}{\e},\frac{1}{\e}\right)}.
	\end{equation*}
	Furthermore, the map
	\begin{equation}
	\label{5.5}
	L_{\e,\p} :~\mathcal{K}_{\e,\p}^\perp\to\mathcal{C}_{\e,\p}^\perp
	\end{equation}
	is surjective.	
\end{proposition}

\begin{proof}
	The proof follows the standard method of Liaypunov-Schmidt reduction which was also used in \cite{gui_1999,gui_2000,wei_2001_gm_2d_weak,wei_2002_gm_2d_strong,wei_2007_existence}. Suppose the proposition is not true. Then there exist sequences $\{\e_k\},~\{\p^k\},~\phi^k$ {satisfying} $\e_k\to0$  {as $k\rightarrow 0$}, $\p^k\in(-1,1)^2$, $\min(|1+p_1^k|,~|1-p_2^k|,~|p_1^k-p_2^k|)>\ov\delta,$  {and }$\phi^k=\phi_{\e_k}\in\mathcal{K}_{\e_k,\p^k}^\perp$ {for all $k\geq 1$} such that
	\begin{equation}
	\label{5.6}
	\|\phi^k\|_{H^{2s}\dd}=1,\qquad
	\|L_{\e_k,\p^k}\phi^k\|_{L^2\dd}\to0,\quad \mathrm{as}\quad k\to\infty.
	\end{equation}
	We define $\phi_{i}^k,~i=1,2$ and $\phi_{3}^k$ as follows:
	\begin{equation}
	\label{5.7}
	\phi_{i}^k(y)=\phi^k(y)\chi\left(\frac{\e y-p_i}{r_0}\right),~i=1,2,\quad
	\phi_{3}^k(y)=\phi^k(y)-\sum_{i=1}^2\phi_{i}^k(y),\quad
	y\in\dd.
	\end{equation}	
	{Although each} $\phi_{i}^k$ is defined only in $\left(-\frac{1}{\e},\frac{1}{\e}\right)$. By a standard result they can be extended to $\R$ such that their norm in $H^{2s}(\R)$ is still bounded by a constant independent of $\e$ and $\p$ for $\e$ small enough. In the following we shall study the corresponding problem in $\R$. To simplify our notation, we keep the same notation for the extension. Since $\{\phi^k_i\}$ is bounded in $H^{2s}_{\mathrm{loc}}(\R)$ it has a weak limit in $H_{\mathrm{loc}}^{2s}(\R)$ and therefore also a strong limit in $L_{\mathrm{loc}}^2(\R)$ and $L_{\mathrm{loc}}^\infty(\R).$ We denote the limit by $\phi_i$. Then $\Phi=\left(\phi_1,\phi_2\right)^T$  solves the system
	$$L\Phi=0.$$
	By Lemma \ref{le3.1}, $\Phi\in \mathrm{Ker}(L)=X_0\oplus X_0$. Since $\phi^k\perp\mathcal{K}_{\e_k,\p^k}^\perp$, by taking $k\to\infty$ we get $\phi\in(X_0\oplus X_0)^\perp$ and therefore $\phi=0$.
	
	By elliptic estimates we get $\|\phi_{i}^k\|_{H^{2s}(\R)}\to0$ as $k\to {\infty}$ for $i=1,2$. Furthermore, $\phi_3^k\to \phi_{3}$ in $H^{2s}(\R)$, where $\phi_{3}$ solves
	\begin{equation}
	\label{5.8}
	(-\Delta)^s\phi+\phi=0\quad \mathrm{in}\quad \R.
	\end{equation}
	Therefore, we conclude $\phi_{3}=0$ and $\|\phi_{3}^k\|_{H^{2s}(\R)}\to0$ as $k\to+\infty.$ This contradicts $\|\phi^k\|_{H^{2s}\left(-\frac{1}{\e_k},\frac{1}{\e_k}\right)}=1$.To complete the proof of Proposition \ref{pr5.1} we just need to show that the operator which is conjugate to $L_{\e,\p}$ (denoted by $L_{\e,\p}^*$) is injective from $\mathcal{K}_{\e,\p}^\perp$ to $\mathcal{C}_{\e,\p}^\perp$. Note that
	$L_{\e,\p}^*=\pi_{\e,\p}\circ \tl_{\e,\p}^*$ with
	$$\tl_{\e,\p}^*\psi=\s_y\psi+\psi-2\frac{\w}{T[\w]}\psi
	+T'[\w]\left(\frac{\w^2}{(T[\w])^2}\psi\right).$$
	The proof for $L_{\e,\p}^*$ follows exactly the same as the one of $L_{\e,\p}$ and we omit the details.
\end{proof}

Now we are in position to solve the problem
\begin{equation}
\label{5.9}
\pi_{\e,\p}^\perp\circ S_{\e}(w_{\e,\p}+\phi)=0.
\end{equation}
Since $L_{\e,\p}\mid_{\mathcal{K}_{\e,\p}^\perp}$ is invertible (call the inverse $L_{\e,\p}^{-1}$) we can rewrite the above problem as
\begin{equation}
\label{5.10}
\phi=-(L_{\e,\p}^{-1}\circ \pi^\perp_{\e,\p}\circ S_{\e}(w_{\e,\p}))-
(L_{\e,\p}^{-1}\circ \pi^\perp_{\e,\p}\circ N_{\e,\p}(\phi))\equiv M_{\e,\p}(\phi),
\end{equation}
where
$$N_{\e,\p}(\phi)=S_\e(w_{\e,\p}+\phi)-S_{\e}(w_{\e,\p})-S_{\e}'(w_{\e,\p})\phi$$
and the operator $M_{\e,\p}$ is defined by $\phi\in H^{2s}\left(-\frac{1}{\e},\frac{1}{\e}\right)$. We are going to show that the operator $M_{\e,\p}$ is a contraction map on
\begin{equation}
\label{5.11}
B_{\e,\sigma}:=\left\{\phi\in H^{2s}\left(-\frac{1}{\e},\frac{1}{\e}\right)\Big|~ \|\phi_{\e}\|_{H^{2s}\left(-\frac{1}{\e},\frac{1}{\e}\right)}<\sigma\right\}
\end{equation}
if $\sigma$ and $\e$ are small enough. We have by the discussion in last section and Proposition \ref{pr5.1} that
\begin{equation}
\label{5.12}
\begin{aligned}
\|M_{\e,\p}(\phi)\|_{H^{2s}\dd}
\leq~& C\left(\|\pi_{\e,\p}^\perp\circ N_{\e,\p}(\phi)\|_{L^2\dd}+\|\pi_{\e,\p}^\perp\circ S_{\e}(w_{\e,\p})\|_{L^2\dd}\right)\\
\leq~&  C(c(\sigma)\sigma+\varepsilon^{2s-1}),
\end{aligned}
\end{equation}
where $C>0$ is a constant independent of $\sigma>0$, $\e>0$ and $c(\sigma)\to0$ as $\sigma\to0$. Similarly we show that
\begin{equation*}
\|M_{\e,\p}(\phi_1)-M_{\e,\p}(\phi_2)\|_{H^{2s}\dd}\leq C(c(\sigma)\sigma)\|\phi_1-\phi_2\|_{H^{2s}\dd},
\end{equation*}
where $c(\sigma)\to 0$ as $\sigma\to0$. If we choose $\sigma=\e^\alpha$ for $\alpha\leq2s-1$ and $\e>0$ sufficiently small then $M_{\e,\p}$ is a contraction map on $B_{\e,\sigma}$. The existence then follows by the standard the fixed point theorem and $\phi_{\e,\p}$ is a solution to \eqref{5.10}. We thus proved

\begin{lemma}
	\label{le5.2}
	There exists $\ove>0,~\overline{\delta}>0$ such that for every pair of $\e,\p$ with $0<\e<\ove$,  $\p\in(-1,1)^2$, and
	$$\min\left\{1+p_1,~1-p_2,~|p_1-p_2|\right\}>\overline\delta,$$
	there is a unique $\phi_{\e,\p}\in\mathcal{K}_{\e,\p}^\perp$ satisfying $S_{\e}(w_{\e,\p}+\phi_{\e,\p})\in\mathcal{C}_{\e,\p}$. Furthermore, we have the estimate
	$$\|\phi_{\e,\p}\|_{H^{2s}\dd}\leq C\e^\alpha$$
	for any $\alpha\leq2s-1$.	
\end{lemma}

More refined estimates for $\phi_{\e,\p}$ are needed. We recall from the discussion in last section that $S[\w]$ can be decomposed into the two parts $S_{1,1}$ and $S_{1,2}$ if $x$ is close to the center of spike, where $S_{1,1}$ is in leading order an odd function and $S_{1,2}$ is in leading order a radially symmetric function. {We can similarly decompose $\phi_{\e,\p}$ as in the following lemma.}

\begin{lemma}
	\label{le5.3}	
	Let $\phi_{\e,\p}$ be defined in Lemma \ref{le5.2}. Then for $x=p_i+\e z$, $|\e z|<\delta$, $i=1,2$, we have the decomposition 
	\begin{equation}
	\label{5.13}
	\phi_{\e,\p}=\phi_{\e,\p,1}+\phi_{\e,\p,2},
	\end{equation}
	where $\phi_{\e,\p,2}$ is an even function in $z$ which satisfies
	\begin{equation}	
	\label{5.14}
	\phi_{\e,\p,2}=O(\e^{2s-1})\quad \mathrm{in}\quad H^{2s}\dd,
	\end{equation}
	and
	\begin{equation}
	\label{5.15}
	\phi_{\e,\p,1}=O(\e)\quad \mathrm{in}\quad H^{2s}\dd.
	\end{equation}
\end{lemma}
\begin{proof}
	We first solve	
	\begin{equation}
	\label{5.16}
	S[\w+\phi_{\e,\p,2}]-S[\w]-\sum_{j=1}^2S_{1,2}\left(y-\frac{p_j}{\e}\right)\in\mathcal{C}_{\e,\p},
	\end{equation}
	for $\phi_{\e,\p,2}\in\mathcal{K}_{\e,\p}^\perp$. Then we solve
	\begin{equation}
	\label{5.17}
	S[\w+\phi_{\e,\p,2}+\phi_{\e,\p,1}]-S[\w+\phi_{\e,\p,2}]-\sum_{j=1}^2S_{1,1}\left(y-\frac{p_j}{\e}\right)\in\mathcal{C}_{\e,\p},
	\end{equation}
	for $\phi_{\e,\p,1}\in\mathcal{K}_{\e,\p}^\perp$. Using the same proof as in Proposition \ref{pr5.1}, both equations \eqref{5.16} and \eqref{5.17} have unique solution provided $\e\ll1$. By uniqueness, $\phi_{\e,\p}=\phi_{\e,\p,1}+\phi_{\e,\p,2}$, and it is easy to see that $\phi_{\e,\p,1}$ and $\phi_{\e,\p,2}$ have the required properties.	
\end{proof}

\medskip
\subsection{The Reduced Problem}
In this subsection, we solve the reduced problem {which will will complete }the proof of  Theorem \ref{th1.exist}. By Proposition \ref{pr5.1} for every $\p\in B_{\e^{2s-1}}(\p^0)$ there exists an unique solution  $\phi_{\e,\p}\in\mathcal{K}_{\e,\p}^\perp$ such that
\begin{equation}
\label{6.1}
S[w_{\e,\p}+\phi_{\e,\p}]=v_{\e,\p}\in\mathcal{C}_{\e,\p}.
\end{equation}
{To complete the proof of Theorem \ref{th1.exist}} we need to determine $\p^\e=(p_1^\e, p_2^\e)$ near $\p^0$ such that  {$S[w_{\e,\p}+\phi_{\e,\p}]\perp \mathcal{C}_{\e,\p}$, which in turn implies that $S[w_{\e,\p}+\phi_{\e,\p}]=0$.} 
To this end, let $W_{\e}:=(W_{\e,1}(\p), W_{\e,2}(\p)):~B_{\e^{2s-1}}(\p^0)\to\mathbb{R}^2$ where
\begin{equation*}
W_{\e,i}(\p):=\e^{-1}\int_{-\frac{1}{\e}}^{\frac{1}{\e}}S[w_{\e,\p}+\phi_{\e,\p}]\frac{\partial w_i}{\partial y}dy,\qquad i=1,2.
\end{equation*} 
Then $W_{\e}(\p)$ is a map which is continuous in $\p$ and our problem is reduced to finding a zero of the vector field $W_{\e}(\p)$. Let us now calculate $W_{\e}(\p)$
\begin{equation}
\label{6.3}
\begin{aligned}
W_{\e,i}(\p)=~&\e^{-1}\int_{-\frac{1}{\e}}^{\frac{1}{\e}}S_{\e}[w_{\e,\p}+\phi_{\e,\p}]\frac{\partial w_i}{\partial y}dy\\
=~&\e^{-1}\int_{-\frac{1}{\e}}^{\frac{1}{\e}}\left[\s (\w+\phi_{\e,\p})+(\w+\phi_{\e,\p})-\frac{(\w+\phi_{\e,\p})^2}{T[\w]+\psi_{\e,\p}}\right]\frac{\partial w_i}{\partial y}dy\\
=~&\e^{-1}\int_{-\frac{1}{\e}}^{\frac{1}{\e}}\left[\s (\w+\phi_{\e,\p})+(\w+\phi_{\e,\p})-\frac{(\w+\phi_{\e,\p})^2}{T[\w]}\right]\frac{\partial w_i}{\partial y}dy\\
&-\e^{-1}\int_{-\frac{1}{\e}}^{\frac{1}{\e}}\left[\frac{(\w+\phi_{\e,\p})^2}{T[\w]+\psi_{\e,\p}}-\frac{(\w+\phi_{\e,\p})^2}{T[\w]}\right]\frac{\partial w_i}{\partial y}dy\\
=~&I_1+I_2,
\end{aligned}
\end{equation}
where $I_1,~I_2$ are defined by the last equality and $\psi_{\e,\p} $ satisifies
\begin{equation}
\label{6.4}
D(-\Delta)^s\psi_{\e,\p}+\psi_{\e,\p}-2c_\e\w\phi_{\e,\p}-c_\e\phi_{\e,\p}^2=0.
\end{equation}
For $I_1$, we have by Lemma \ref{le5.3}
\begin{equation}
\label{6.5}
\begin{aligned}
I_1=~&\e^{-1}\left(\int_{-\frac{1}{\e}}^{\frac{1}{\e}}\left[\s(\w+\phi_{\e,\p})+(w_{\e,\p}+\phi_{\e,\p})-\frac{(\w+\phi_{\e,\p})^2}{T[\w](p_i)}\right]\frac{\partial w_i}{\partial y}dy\right.\\
&\left.+\int_{-\frac{1}{\e}}^{\frac{1}{\e}}\frac{(\w+\phi_{\e,\p})^2}{(T[\w](p_i))^2}(T[\w](p_i+\e y)-T[\w](p_i))\frac{\partial w_i}{\partial y}dy\right)+O(\e^{2s-1})\\
=~&\e^{-1}\left(\int_{-\frac{1}{\e}}^{\frac{1}{\e}}\left[\s(\hat\xi_iw_i+\phi_{\e,\p})+(\hx_iw_i+\phi_{\e,\p})-\frac{(\hx_iw_i+\phi_{\e,\p})^2}{T[\w](p_i)}\right]\frac{\partial w_i}{\partial y}dy\right)\\
&+\e^{-1}\left(\int_{-\frac{1}{\e}}^{\frac{1}{\e}}
\frac{(\hx_iw_i+\phi_{\e,\p,2})^2}{(T[\w](p_i))^2}
(T[\w](p_i+\e y)-T[\w](p_i))\frac{\partial w_i}{\partial y}dy
\right)+O(\e^{2s-1}).	
\end{aligned}
\end{equation}
Note that, by Lemma \ref{le5.3}, we have
\begin{equation}
\label{6.6}
\int_{-\frac{1}{\e}}^{\frac{1}{\e}}[\s\phi_{\e,\p}+\phi_{\e,\p}-2w_{i}\phi_{\e,\p}]\frac{\partial w_i}{\partial y}dy
=\int_{-\frac{1}{\e}}^{\frac{1}{\e}}\phi_{\e,\p,1}\frac{\partial}{\partial y}\left(\s w_i+w_i-w_i^2\right)dy+O(\e^{1+2s})=O(\e^{1+2s}),
\end{equation}
and
\begin{equation}
\label{6.7}	
\int_{-\frac{1}{\e}}^{\frac{1}{\e}}\phi_{\e,\p}^2\frac{\partial w_i}{\partial y}dy
=\int_{-\frac{1}{\e}}^{\frac{1}{\e}}\phi_{\e,\p,1}\phi_{\e,\p,2}\frac{\partial w_i}{\partial y}dy+h.o.t.=O(\e^{2s}).
\end{equation}
Now by Lemma \ref{le5.3} and equations \eqref{6.5} and \eqref{6.6}  we have
\begin{equation}
\label{6.8}
\begin{aligned}
I_1=~&\e^{-1}\int_{-\frac{1}{\e}}^{\frac{1}{\e}}w_i^2(T[\w](p_i+\e z)-T[\w](p_i))\frac{\partial w_i}{\partial y}dy+O(\e^{2s-1})\\	
=~&\e^{-1}\int_{-\frac{1}{\e}}^{\frac{1}{\e}}	
w_i^2\left(P_i(z)+\e\sum_{j\neq i}\hat\xi_j^2z\nabla_{p_i}G_D(p_i,p_j)\right)\frac{\partial w_i}{\partial y}dy+O(\e^{2s-1})\\
=~&-\frac13\int_{\R}w^3(y)dy\sum_{j\neq i}\hat\xi_j^2\nabla_{p_i}G_D(p_i,p_j)+O(\e^{2s-1}).	
\end{aligned}
\end{equation}
Similarly, we calculate
\begin{equation}
\label{6.9}
\begin{aligned}
I_2=~&\e^{-1}\int_{-\frac{1}{\e}}^{\frac{1}{\e}}
\left[\frac{(\w+\phi_{\e,\p})^2}{T[\w]+\psi_{\e,\p}}-\frac{(\w+\phi_{\e,\p})^2}{T[\w]}\right]\frac{\partial w_i}{\partial y}dy\\
=~&-\e^{-1}\int_{-\frac{1}{\e}}^{\frac{1}{\e}}
\frac{(\w+\phi_{\e,\p})^2}{(T[\w])^2}\psi_{\e,\p}\frac{\partial w_i}{\partial y}dy+O(\e^{2s-1})\\
=~&-\e^{-1}\hx_i\int_{-\frac{1}{\e}}^{\frac{1}{\e}}\frac13\frac{\partial w^3}{\partial y}(\psi_{\e,\p}-\psi_{\e,\p}(p_i))dy+O(\e^{2s-1}).
\end{aligned}
\end{equation}
Since $\psi_{\e,\p}$ satisifies \eqref{6.4},  {a similar argument to that used in Lemma \ref{le5.3} gives}
\begin{equation}
\label{6.10}
\begin{aligned}	
\psi_{\e,\p}(p_i+\e z)-\psi_{\e,\p}(p_i)
=~&c_\e\int_{-1}^{1}\left(G_D(p_i+\e z,\zeta)-G_D(p_i,\zeta)\right)\left(2\w\left(\frac{\zeta}{\e}\right) \phi_{\e,\p}\left(\frac{\zeta}{\e}\right)+\phi^2_{\e,\p}\left(\frac{\zeta}{\e}\right)\right)d\zeta\\
=~&o\left(\e\sum_{j\neq i}\hx_j^2z\nabla_{p_i}G_D(p_i,p_j)\right)+\hat P_i(z)+h.o.t.,
\end{aligned}
\end{equation}
where $\hat P_i(z)$ is an even function in $z=y-\frac{p_i}{\e}$. Substituting \eqref{6.10} into \eqref{6.9} we obtain that
\begin{equation}
\label{6.11}
I_2=o( \sum_{j\neq i} \hx_j^2\nabla_{p_i}G_D(p_i,p_j))+o(\e^{2s-1}).
\end{equation}
Combining the estimates for $I_1$ and $I_2$, we obtain
\begin{equation}
\label{6.12}
W_{\e,i}(\p)=-\frac13 \int_{\R}w^3(y)dy\sum_{j\neq i}\hx_j^2\nabla_{p_i}G_D(p_i,p_j)(1+o(1))+O(\e^{2s-1})
=-\frac13F_i(\p)\int_{\R}w^3(y)dy +O(\e^{2s-1}),
\end{equation}
where $F_i(\p)$ is defined in \eqref{2.f}. {From (H3) we have $F(\p^0)=0$ and from numerial calculations of the Green's function we deduce that $p_2^0-p_1^0=1$.}  By symmetry we conclude that if there exists $\p =(p_1 ,p_2 )$ such that either one of  {$W_{\e,1}(\p )=0$ or $W_{\e,2}(\p )=0$} then $W_\e(\p )=0.$ For $W_{\e,i}$ we have
\begin{align*}
W_{\e,i}(p)=&-\frac13\int_{\R}w^3(y)dy\left((p_1-p_1^0)\hx_2^2\nabla_{p_1}\nabla_{p_1}G_D(p_1^0,p_2^0)
+(p_2-p_2^0)\hx_2^2\nabla_{p_2}\nabla_{p_1}G_D(p_1^0,p_2^0)\right)\\
&+O(|\p-\p^0|^2+\e^{2s-1}).
\end{align*}
By assumption (H3) we have $\nabla_{p_1}\nabla_{p_1}G_D(p_1^0,p_2^0)\neq 0$. As a consequence, we can  apply Brouwer's fixed point theorem to show that for $\e\ll1$ there exists a point $\p^\e$ such that $W_{\e}(\p^\e)=0$ and $\p^\e\in B_{\e^{2s-1}}(\p^0).$ Thus we have proved the following proposition

\begin{proposition}
	\label{pr6.1}
	For $\e$ sufficiently small there exist points $\p^\e$ with $\p^\e\to\p^0$ such that $W_{\e}(\p^\e)=0.$
\end{proposition}

\begin{proof}[Proof of Theorem \ref{th1.exist}.]
	By above Proposition, there exists $\p^\e\to\p^0$ such that $W_\e(\p^\e)=0$. In other words, $S[w_{\e,\p^\e}+\phi_{\e,\p^\e}]=0$. Let $u_{\e}=c_\e(w_{\e,\p^\e}+\phi_{\e,\p^\e})$, $v_{\e}=c_{\e}T[w_{\e,\p^\e}+\phi_{\e,\p^\e}]$. By the Maximum principle, $u_\e>0$ and $v_\e>0$. Moreover $(u_\e,v_\e)$ satisfies all the properties of Theorem \ref{th1.exist}.
\end{proof}

\vspace{0.5cm}
\section{Rigorous proof of the stability analysis}\label{sec:proof-stability}
{The linear stability of the two-spike solution constructed above is determined by two classes of eigenvalues: the large and small eigenvalues satisfying $\lambda_\varepsilon=O(1)$ and $\lambda_\varepsilon\rightarrow 0$ as $\varepsilon\rightarrow0$ respectively. In the following two subsections we consider each case separately.} 
\subsection{Stability Analysis: Large Eigenvalues}
In this subsection, we consider the stability of the steady state $(u_\e,v_\e)$ constructed in Theorem \ref{th1.exist}. Linearizing around the equilibrium states,
\begin{equation}
\label{7.1}
u=u_\e+\phi_\e(x)e^{\lambda_\e t},\quad v=v_{\e} +\psi_{\e}e^{\lambda_\e t}=T[u_{\e}]+\psi_{\e}e^{\lambda_\e t},		
\end{equation}
and substituting the result into (GM) we deduce the following eigenvalue problem
\begin{equation}
\label{7.2}
\begin{cases}
\s_y\phi_\e+\phi_\e-2\frac{u_\e}{T[u_\e]}\phi_\e+\frac{u_\e^2}{(T[u_\e])^2}\psi_\e+\lambda_\e\phi_\e=0,\\
D\s\psi_\e+\psi_\e-2c_\e u_\e\phi_\e+\tau\la_\e\psi_\e=0,
\end{cases}
\end{equation}
where $\la_\e$ is some complex number. In this section, we study the large eigenvalues, i.e. {those for which we may} assume that there exists $c>0$ such that $|\la_\e|\geq c>0$ for $\e$ small. If $\Re(\la_\e)<-c$ then  we are done (since these eigenvalues are always stable) and we therefore assume that $\Re(\la_\e)\geq-c$. {For a subsequence $\varepsilon\rightarrow 0$ and $\lambda_\varepsilon\rightarrow \lambda_0$ we shall derive a limiting NLEP satisfied by $\lambda_0$.}

We first present the case $\tau=0$. At the end, we shall explain how we proceed  {when $\tau>0$ is sufficiently small}. By the second equation of \eqref{7.2}, we have $\psi_\e=T'[u_\e](\phi_\e)$. Let us assume that  {$\|\phi_\e\|_{H^{2s}\dd}=1$ and cut off 
	We cut off $\phi_\e$ as follows:
	\begin{equation}
	\label{7.4}	
	\phi_{\e,i}(y)=\phi_\e(y)\chi(\frac{\e y-p_i^\e}{r_0}),
	\end{equation}
	where $\chi(x)$ is a given in \eqref{4.1} of \S\ref{sec:proof-existence}. Using Lemma \ref{le5.3}  {together with}  $\Re(\la_\e)\geq-c$, the asymptotic expansion of $u_\e$ given in Theorem \ref{th1.exist}, and the algebraic decay of $w$  {given in Proposition \ref{pr3.1}}, we get that
	\begin{equation}
	\label{7.5}
	\phi_\e=\sum_{i=1}^2\phi_{\e,i}+h.o.t.. \qquad \mathrm{in}~H^{2s}\left(-\frac{1}{\e},\frac{1}{\e}\right).
	\end{equation}
	Then by standard procedure we extend $\phi_{\e,i}$ to a function defined on $\R$ such that
	\begin{equation}
	\label{7.6}	
	\|\phi_{\e,i}\|_{H^{2s}(\R)}\leq C\|\phi_{\e,i}\|_{H^{2s}\dd},\quad i=1,2.	
	\end{equation}
	Without loss of generality we may assume that $\|\phi_\e\|_{H^{2s}(\R)}=1$ and by taking a subsequence of $\e$, we may also assume that $\phi_{\e,i}\to\phi_i$ {strongly }as $\e\to0$ in $L^2\cap L^\infty$ for $i=1,2$, on compact subsets of $\R$. Therefore we also have
	\begin{equation}
	\label{7.7}
	w\phi_{\e,i}\to w\phi_i~\ \mathrm{as}~\ \e\to0, \quad \mathrm{strongly~in}~L^\infty(\R).
	\end{equation}
	
	It is known that
	\begin{equation}
	\label{7.8}
	\psi_\e(x)=c_\e\int_{-1}^1G_D(x,\zeta)u_\e\left(\frac{\zeta}{\e}\right)\phi_\e\left(\frac{\zeta}{\e}\right)d\zeta.
	\end{equation}
	Now we use the expansion of $u_\e$ to calculate the value of $\psi_\e$ at $x=p_i^\e$ {for each $i=1,2$}
	\begin{equation}
	\label{7.9}
	\begin{aligned}
	\psi_\e(p_i^\e)=~&2c_\e\int_{-1}^1G_D(p_i^\e,\zeta)\sum_{j=1}^2\hx_jw\left(\frac{\zeta-p_j^\e}{\e}\right)\chi\left(\frac{\e\zeta-p_j}{r_0}\right)\phi_\e\left(\frac{\zeta}{\e}\right)d\zeta+h.o.t.\\
	=~&2\e c_\e \sum_{j=1}^2\hat\xi_jG_D(p_i,p_j)\int_{\R}w\phi_jdy+o_\e(1).
	\end{aligned}	
	\end{equation}
	Substituting \eqref{7.9} into the first equation of \eqref{7.2} and letting $\e\to0$, we obtain the nonlocal eigenvalue problem
	\begin{equation}
	\label{7.10}
	\s\phi_i+\phi_i-2w\phi_i+2\left(\int_{\R}w^2(y)dy\right)^{-1}\left(\int_{\R}\sum_{j=1}^2\hat\xi_jG_D(p_i,p_j)w\phi_jdy\right)w^2+\la_0\phi_i=0,\quad i=1,2.
	\end{equation}
	We can rewrite \eqref{7.10} in matrix form as
	\begin{equation}
	\label{7.11}
	\s\Phi+\Phi-2w\Phi+2\left(\int_{\R}w^2(y)dy\right)^{-1}\left(\int_{\R}w\mathcal{B}\Phi dy\right)w^2+\la_0\Phi=0
	\end{equation}
	where $\mathcal{B}$ is the matrix introduced in \eqref{2.matrixB} and {$\Phi=(\phi_1,\phi_2)^T\in (H^{2s}(\R))^2$. We then have the following conclusion} 
	
	\begin{theorem}
		\label{th7.1}
		Let $\la_\e$ be an eigenvalue of \eqref{7.2} such that $\Re(\la_\e)>-c$ for some $c>0.$
		\begin{enumerate}
			\item [(1)] Suppose that for suitable sequences $\e_n\to0$ we have $\la_{\e_n}\to\la_0\neq 0$. Then $\la_0$ is an eigenvalue of the problem given in \eqref{7.10}.
			
			\item [(2)] Let $\la_0\neq 0$ with $\Re(\lambda_0)>0$ be an eigenvalue of the problem given in \eqref{7.10}. Then for $\e$ sufficiently small, there is an eigenvalue $\la_\e$ of \eqref{7.1} with $\la_\e\to\la_0$ as $\e\to0.$
		\end{enumerate}
	\end{theorem}

	\begin{proof}
		The proof of (1) follows {from a similar asymptotic analysis to that used in \S\ref{sec:proof-existence}}.
		
		To prove part (2) of Theorem \ref{th7.1} we follow the argument given in \S2 of \cite{dancer_2001_hopf}.  {We assume that $\la_0\neq 0$ is an eigenvalue of problem \eqref{7.10} with $\Re(\la_0)>0$ and} we first note that from the equation for $\psi_\e$, we can express $\psi_\e$ in terms of $\phi_\e$ as in \eqref{7.8}. {Then we} rewrite the first equation \eqref{7.2} as  
		\begin{equation*}
		\phi_\e=-R_{\e}(\la_\e)\left[2\frac{u_\e\phi_\e}{v_\e}-\frac{u_\e^2}{v_{\e}^2}\psi_{\e}\right],
		\end{equation*}	
		where $R_\e(\la_{\e})$ is the inverse of $(-\Delta)^s+(1+\la_{\e})$ in $H^{2s}(\R)$ and $\psi_{\e}=T'_\e[u_\e](\phi_{\e})$ is given in the second equation of \eqref{7.2}. The {key observation} is that $R_\e(\la_\e)$ is a Fredholm type operator if $\e$ is sufficiently small. The rest of the argument follows as in \cite{dancer_2001_hopf}.	
	\end{proof}

	{By diagonalizing $\mathcal{B}$} we see that the eigenvalue problem \eqref{7.11} can be reduced to the nonlocal eigenvalue problems
	\begin{equation}
	\label{7.12}
	\s\hat\phi_i+\hat\phi_i-2w\hat\phi_i+2\sigma_i\frac{\int_{\R}w\hat\phi_idy}{\int_{\R}w^2dy}w^2+\lambda_0\hat\phi_i=0,\quad \hat\phi_i\in H^{2s}(\R), \quad i=1,2,	
	\end{equation}
	where {$\sigma_1$ and $\sigma_2$} are the two eigenvalues of $\mathcal{B}$.
	
	We now study the stability of \eqref{7.2} for large eigenvalues explicitly. Suppose that 
	\begin{equation}
	\label{7.13}	
	2\min_{\sigma\in\lambda(\mathcal{B})}\sigma<1.
	\end{equation}
	Then by Theorem \ref{th3.stability}-{(1)} there exists an {unstable } eigenvalue of \eqref{7.11} and  {therefore} by Theorem \ref{th7.1} there exists an eigenvalue $\la_\e$ of \eqref{7.2} such that $\Re(\la_\e)>c_0$ for some positive number $c_0$. This implies that $(u_\e,v_\e)$ is unstable. {On the other hand if $2\min_{\sigma\in\lambda(\mathcal{B})}\sigma>1$} {then} by Theorem \ref{th3.stability}-{(2)} any nonzero eigenvalue $\la_0$ is stable. Therefore by Theorem \ref{th7.1} for $\e$ small enough all nonzero eigenvalues $\la_\e$ of \eqref{7.2} for which $|\la_\e|\geq c>0$ holds, satisfy $\Re(\la_\e)\leq -c<0$ for $\e$ small enough. 
	
	{Finally we comment that when $\tau\neq 0$ and $\tau$ is small.} We shall apply the results of Theorem \ref{th3.2}. In this case, the matrix $\mathcal{B}$ will have to be replaced by the matrix $\mathcal{B}_{\tau\la_\e}$ which depends on $\tau\e$. In particular the Green's function $G_D$ is replaced by the { Green's function $G_D^\lambda$ satisfying}
	\begin{equation}
	\label{7.green}
	D\s G_D^\lambda+(1+\tau\la_\e)G_D^\lambda=\delta_z,\quad G_D^\lambda(x+2,z)=G_D^\lambda(x,z).
	\end{equation}
	It is {then} easy to check that {the eigenvalues of $\mathcal{B}_{\tau\la_\e}$ satisfy the same properties as those of $\mathcal{B}$} provided that $\tau$ is sufficiently small.
		
	\subsection{Stability Analysis: Small Eigenvalues}
	We now study the eigenvalue problem \eqref{7.2} with respect to small eigenvalues. Namely, we assume that $\lambda_\e\to0$ as $\e\to0$. Let
	\begin{equation}
	\label{8.1}
	\bar{u}_\e=w_{\e,\p^\e}+\phi_{\e,\p^\e},\quad
	\bar{v}_\e=T[w_{\e,\p^\e}+\phi_{\e,\p^\e}],
	\end{equation}
	where $\p^\e=(p_1^\e,p_2^\e)$. After rescaling, the eigenvalue problem \eqref{7.2} becomes
	\begin{equation}
	\label{8.2}
	\begin{cases}
	\s_y\phi_\e+\phi_\e-\frac{2\ow_\e}{\oh_\e}\phi_\e+\frac{\ow_\e^2}{\oh_\e^2}\psi_\e+\lambda_\e\phi_\e=0,\\
	D\s\psi_\e+\psi_\e-2c_\e\ow_\e\phi_\e+\tau\lambda_\e\psi_\e=0,
	\end{cases}		
	\end{equation}
	where $c_\e$ is given by \eqref{4.4}. We take $\tau=0$ for simplicity. As $\tau\lambda_\e\ll1$ the results in this section are also valid for $\tau$ finite, this is due to the fact that the small eigenvalue are of the order $O(\e^2)$, we shall prove it in this subsection.
	
	We cut off $\ow_\e$ as follows
	\begin{equation}
	\label{8.3}
	\tw_{\e,i}(y)=\chi\left(\frac{\e y-p_i^\e}{r_0}\right)\ow_\e(y),\quad i=1,2,
	\end{equation}
	where $\chi(x)$ and $r_0$ are given in \S\ref{sec:proof-existence}. Similarly to the \S\ref{sec:proof-existence} we define
	\begin{gather*}
	\mathcal{K}_{\e,\p,new}:=~\mathrm{Span}\left\{\tw'_{\e,i}\mid i=1,2\right\}\subset H^{2s}\dd,\\
	\mathcal{C}_{\e,\p,new}:=~\mathrm{Span}\left\{\tw'_{\e,i}\mid i=1,2\right\}\subset L^2\dd.
	\end{gather*}
	Then it is easy to see that
	\begin{equation}
	\label{8.4}
	\ow_{\e}(y)=\sum_{i=1}^2\tw_{\e,i}(y)+h.o.t..	
	\end{equation}
	Note that
	$$\tw_{\e,i}(y)\sim\hat\xi_iw\left(y-\frac{p_i^\e}{\e}\right) \quad \mathrm{in}\quad H^{2s}(-1,1)$$
	and $\tilde u_{\e,i}$ satisfies
	\begin{equation}
	\label{8.5}
	\s\tw_{\e,i}+\tw_{\e,i}-\frac{\tw_{\e,i}^2}{\oh_\e}+h.o.t.=0.	
	\end{equation}
	Thus $\tw_{\e,i}':=\frac{d\tw_{\e,i}}{dy}$ satisfies
	\begin{equation}
	\label{8.6}
	\s_y\tw'_{\e,i}+\tw'_{\e,i}-2\frac{\tw_{\e,i}}{\oh_\e}\tw'_{\e,i}+\e\frac{\tw_{\e,i}^2}{\oh_\e^2}\oh_\e'+h.o.t.=0,
	\end{equation}
	and we have
	\begin{equation*}
	\tw'_{\e,i}=\hat\xi_i\frac{\partial w}{\partial y}\left(y-\frac{p_i^\e}{\e}\right)(1+o(1)).
	\end{equation*}
	Let us now decompose
	\begin{equation}
	\label{8.7}
	\phi_\e=\sum_{i=1}^2a_i^\e\tilde u_{\e,i}'+\phi_\e^\perp,
	\end{equation}
	where $a_i^\e$ are complex numbers and $\phi_\e^\perp\perp\mathcal{K}_{\e}$. Similarly, we can decompose
	\begin{equation}
	\label{8.8}
	\psi_\e=\sum_{i=1}^2a_i^\e\psi_{\e,i}+\psi_\e^\perp,
	\end{equation}
	where $\psi_{\e,i}$ satisfies
	\begin{equation}
	\label{8.9}
	D\s\psi_{\e,i}+\psi_{\e,i}-2c_\e\bar{u}_\e\tw_{\e,i}'=0,\quad i=1,2,
	\end{equation}
	and $\psi_{\e}^\perp$ satisfies
	\begin{equation}
	\label{8.10}
	D\s\psi_\e^\perp+\psi_\e^\perp-2c_\e\bar{u}_\e\phi_\e^\perp=0.
	\end{equation}
	We impose periodic boundary conditions in both of these equations.
	
	Suppose that $\|\phi_\e\|_{H^{2s\dd}}=1$. Then $|a_i^\e|\leq C$ since
	\begin{equation*}
	a_i^\e=\dfrac{\int_{-\frac{1}{\e}}^{\frac{1}{\e}}\phi_\e\frac{\partial \tw_{\e,i}}{\partial y}dy}{\hat\xi_i^2\int_{\R}w^2dy}+o(1).
	\end{equation*}
	Substituting the decompositions of $\phi_\e$ and $\psi_\e$ into \eqref{8.2} we have
	\begin{equation}
	\label{8.11}
	\begin{aligned}
	\s_y\phi_\e^\perp+\phi^\perp_\e-\frac{2\ow_\e}{\oh_\e}\phi_\e^\perp+\frac{\ow_\e^2}{\oh_\e^2}\psi_\e^\perp+\lambda_\e\phi_\e^\perp-\e\sum_{i=1}^2a_i^\e\left(\frac{\tw_{\e,i}^2}{\oh_\e^2}\oh_\e'-\frac{1}{\e}\frac{\ow_\e^2}{\oh_\e^2}\psi_{\e,i}\right)+h.o.t.=-\lambda_\e\sum_{i=1}^2a_i^\e\tw_{\e,i}'.
	\end{aligned}
	\end{equation}	
	Let us first compute
	\begin{equation*}
	\begin{aligned}
	J_0:=~&\e\sum_{i=1}^2a_i^\e\left(\frac{\tw_{\e,i}^2}{\oh^2_{\e}}\oh'_\e-\frac{1}{\e}\frac{\ow_\e^2}{\oh_\e^2}\psi_{\e,i}\right)\\	
	=~&\e\sum_{i=1}^2a_i^\e\left(\frac{\tw_{\e,i}^2}{\oh_\e^2}(\oh'_\e-\frac{1}{\e}\psi_{\e,i})\right)-\sum_{i=1}^2a_i^\e\sum_{j\neq i}\frac{\tilde u_{\e,j}^2}{\oh^2_\e}\psi_{\e,i}+h.o.t.\\	
	=~&\e\sum_{i=1}^2a_i^\e\frac{\tw_{\e,i}^2}{\oh_{\e}^2}\left(-\frac{1}{\e}\psi_{\e,i}+\oh'_\e\right)-\sum_{i=1}^2\sum_{j\neq i}a_j^{\e}\psi_{\e,j}\frac{\tw_{\e,i}^2}{\oh_\e^2}+h.o.t..	
	\end{aligned}	
	\end{equation*}
	We rewrite $J_0$ as follows
	\begin{equation}
	\label{8.12}
	J_0=-\e\sum_{i=1}^2\sum_{j=1}^2a_j^\e\frac{\tw_{\e,i}^2}{\oh_\e^2}\left(\frac{1}{\e}\psi_{\e,j}-\oh'_\e\delta_{ij}\right)+h.o.t..
	\end{equation} 	
	Let us also set
	\begin{equation}
	\label{8.13}
	\tilde L_\e\phi_\e^\perp:=-\s_y\phi_\e^\perp-\phi_\e^\perp
	+2\frac{\ow_\e}{\oh_\e}\phi_\e^\perp-\frac{\ow_\e^2}{\oh_\e^2}\psi_\e^\perp
	\end{equation}
	and
	\begin{equation}
	\label{8.14}
	{\bf a}_\e:=(a_1^\e,a_2^\e)^T.
	\end{equation}
	Multiplying both sides of \eqref{8.11} by $\tw_{\e,l}'$ and integrating over $\dd$, we obtain
	\begin{equation}
	\label{8.15}
	r.h.s.=-\la_\e\sum_{i=1}^2a_i^\e\int_{-\frac{1}{\e}}^{\frac{1}{\e}}\tw'_{\e,i}\tw'_{\e,l}dy=-\la_\e a_l^\e\hat\xi_l^2\int_{\R}(w'(y))^2dy(1+O(\e^{2s+1})),
	\end{equation}
	and
	\begin{equation}
	\label{8.16}
	\begin{aligned}
	l.h.s.=~&\left(\sum_{i=1}^2\sum_{j=1}^2a_j^\e\int_{-\frac{1}{\e}}^{\frac{1}{\e}}\frac{\tw_{\e,i}^2}{\oh_\e^2}(\psi_{\e,j}-\e\oh'_\e\delta_{ij})
	\tw'_{\e,l}dy+\int_{-\frac{1}{\e}}^{\frac{1}{\e}}\frac{\tw_{\e,l}^2}{\oh_\e^2}\psi_\e^\perp \tw'_{\e,l}dy
	-\e\int_{-\frac{1}{\e}}^{\frac{1}{\e}}\frac{\tw_{\e,l}^2}{\oh_\e^2}\oh'_\e\phi_\e^\perp dx \right)(1+o(1))\\
	=~&(J_{1,l}+J_{2,l}+J_{3,l})(1+o(1)),
	\end{aligned}
	\end{equation}
	where $J_{i,l},~i=1,2,3,$ are defined by the last inequality.
	
	We define the vectors
	\begin{equation}\label{eq:J-def}
	\mathbf{J}_i=(J_{i,1},J_{i,2})^T,\quad i=1,2,3.
	\end{equation}
	To give estimates on each $\mathbf{J}_i$ ($i=1,2,3$) we need the following three lemmas.
	
	\begin{lemma}
		\label{le8.1}
		We have
		\begin{equation}
		\label{8.17}
		(\psi_{\e,j}-\e\oh'_\e \delta_{ji})(p_i^\e)=
		-\e\nabla\mathcal{G}_D^T\mathcal{H}^2+O(\e^2).
		\end{equation}
	\end{lemma}
	
	\begin{proof}
		Note that for $i\neq j$, we have
		\begin{equation}
		\label{8.18}
		\begin{aligned}
		(\psi_{\e,j}-\e\oh'_{\e}\delta_{ji})(p_i^\e)=\psi_{\e,j}(p_i^\e)=2c_\e\int_{-1}^1G_D(p_i^\e,\zeta)\ow_{\e}\tw'_{\e,j}d\zeta
		=-\e\hx_j^2\nabla_{p_j^\e}G_D(p_i^\e,p_j^\e) +O(\e^{1+2s}).
		\end{aligned}	
		\end{equation}	
		Next we compute $\psi_{\e,i}-\oh'_\e$ near $p_i^\e$:
		\begin{equation*}
		\begin{aligned}
		\bar v_{\e}(x)=c_\e\int_{-1}^1G_D(x,\zeta)\ow^2\left(\frac{\zeta}{\e}\right)d\zeta	
		=\e c_\e\int_{-\infty}^{\infty}G_D(x,\e z)\tw_{\e,i}^2(z)dz+
		c_\e\sum_{j\neq i}\int_{-1}^1G_D(x,\zeta)\tw^2_{\e,j}\left(\frac{\zeta}{\e}\right)d\zeta+O(\e^{1+2s}).	
		\end{aligned}	
		\end{equation*}	
		So
		\begin{equation*}
		\begin{aligned}
		\oh'_\e(x)=~&\e c_\e\int_{-\infty}^\infty \nabla_xG_D(x,\e z)\tw_{\e,i}^2(z)dz+c_\e\sum_{j\neq i}\int_{-1}^1\nabla_xG_D(x,\zeta)\tw^2_{\e,j}\left(\frac{\zeta}{\e}\right)d\zeta+O(\e^{1+2s}).	
		\end{aligned}	
		\end{equation*}
		Thus
		\begin{equation*}
		\begin{aligned}
		\psi_{\e,i}(x)-\e\oh'_\e(x)=~&2\e c_\e\int_{-\infty}^{\infty}G_D(x,\e z)\tw_{\e,i}\tw'_{\e,i}dz-\e^2c_\e\int_{-\infty}^\infty \nabla_xG_D(x,\e z)\tw^2_{\e,i}(z)dz\\
		&-\e c_\e\sum_{j\neq i}\int_{-1}^1\nabla_xG_D(x,\zeta)\tw^2_{\e,j}\left(\frac{\zeta}{\e}\right)d\zeta
		+O(\e^{1+2s}).	
		\end{aligned}		
		\end{equation*}
		Therefore we have,
		\begin{equation}
		\label{8.19}
		\begin{aligned}
		\psi_{\e,i}(p_i^\e)-\e\oh'_\e(p_i^\e)=~&2\e c_\e\int_{-\infty}^{\infty}G_D(p_i^\e,\e z)\tw_{\e,i}\tw'_{\e,i}dz-\e^2c_\e\int_{-\infty}^\infty \nabla_{p_i^\e}G_D(p_i^\e,\e z)\tw^2_{\e,i}(z)dz\\
		&-\e^2 c_\e\sum_{j\neq i}\int_{-\infty}^\infty\nabla_{p_i^\e}G_D(p_i^\e,\e z)\tw_{\e,j}^2(z)dz+O(\e^{1+2s})\\
		=&-2\e\hx_i^2\nabla_{p_i^\e}G_D(p_i^\e,p_i^\e)-\e\sum_{j\neq i}\hx_j^2\nabla_{p_i^\e}G_D(p_i^\e,p_j^\e)+o(\e^{2})\\
		=&-\e\hx_i^2\nabla_{p_i^\e}G_D(p_i^\e,p_i^\e)+o(\e^2).	
		\end{aligned}	
		\end{equation}
		Equation \eqref{8.17} then follows from solving \eqref{8.18} and \eqref{8.19}.
	\end{proof}
	
	\begin{lemma}
		\label{le8.2}	
		{Let $q_{ji}$ be defined as in \eqref{2.q}. Then we have}
		\begin{equation}
		\label{8.20}
		\begin{aligned}
		(\psi_{\e,i}-\e\oh'_\e\delta_{ji})(p_j^\e+\e z)-(\psi_{\e,i}-\e\oh'_\e\delta_{ji})(p_j^\e)=
		-\e^2 z\left(\nabla_{p_j^\e}\nabla_{p_i^\e}G_D(p_j^\e,p_i^\e)
		+q_{ji}\delta_{ji}\right)\hat\xi_i^2+o(\e^2).
		\end{aligned}		
		\end{equation}	
	\end{lemma}
	
	We next study the asymptotic expansion of $\phi_\e^\perp$. Let us first {define}
	\begin{equation}
	\label{8.21}
	\phi_{\e,i}^1=\sum_{j=1}^2\nabla_{p_i^\e}\hat\xi_jw\left(y-\frac{p_j^\e}{\e}\right), \quad \phi_\e^1:=-\e\sum_{i=1}^2a_i^\e\phi_{\e,i}^1.
	\end{equation}
	Then we have {the following lemma}.
	\begin{lemma}
		\label{le8.3}
		Let be $\e$ sufficiently small. Then
		\begin{equation}
		\label{8.22}
		\|\phi_\e^\perp-\phi_\e^1\|_{H^{2s}\dd}=o(\e).
		\end{equation}	
	\end{lemma}

	\begin{proof}
		We first derive a relation between $\psi_\e^\perp$ and $\phi_\e^\perp$. Note that similar to the proof of Proposition \ref{pr5.1}, $\tilde{L}_\e$ is invertible from $\mathcal{K}_{\e,\p,new}$ to $\mathcal{C}_{\e,\p,new}$. By Lemma \ref{le8.1} and the fact that $\tilde L_\e$ is invertible, we deduce that
		\begin{equation}
		\label{8.23}
		\|\phi_\e^\perp\|_{H^{2s}\dd}=O(\e).
		\end{equation}	
		Let us decompose
		\begin{equation}
		\label{8.24}
		\tilde \phi_{\e, i}=\frac{\phi_\e^\perp}{\e}\chi\left(\frac{\e y-p_i^\e}{r_0}\right).	
		\end{equation}		
		Then
		\begin{equation}
		\label{8.25}
		\phi_\e^\perp=\e\sum_{i=1}^2\tilde\phi_{\e,i}+h.o.t.
		\end{equation}
		Suppose that
		\begin{equation}
		\label{8.26}
		\tilde\phi_{\e,i}\to\phi_i\quad \mathrm{in}\quad H^1\dd.
		\end{equation}	
		{By} the equation for $\psi_\e^\perp$ (similar to the proof of Lemma \ref{le8.1})
		\begin{equation}
		\label{8.27}
		\begin{aligned}
		\psi_\e^\perp(p_i^\e)=2\e c_\e\sum_{j=1}^2\int_{-1}^1G_D(p_i^\e,z)\ow_\e\tilde\phi_{\e,j}(z)dz+o(\e)
		=2\e\sum_{j=1}^2G_D(p_i^\e,p_j^\e)\hat\xi_j\frac{\int_{\R}w\phi_j dx}{\int_{\R}w^2dx}+o(\e),
		\end{aligned}	
		\end{equation}
		{and therefore}
		\begin{equation}
		\label{8.28}
		(\psi_\e^\perp(p_1^\e),\psi_\e^\perp(p_2^\e))^T
		=2\e\mathcal{G}_D\mathcal{H}\frac{\int_{\R}w\Phi_0dx}{\int_{\R}w^2dx}+o(\e),
		\end{equation}	
		where {$\Phi_0=(\phi_1,\phi_2)^T$}. Substituting \eqref{8.28} into \eqref{8.11} and using Lemma \ref{le8.1}, we have that $\Phi_0$ satisfies	
		\begin{equation}
		\label{8.29}
		\begin{aligned}
		(-\Delta)^s\Phi_0+\Phi_0-2w\Phi_0+2\mathcal{G}_D\mathcal{H}\frac{\int_{\R}w\Phi_0dx}{\int_{\R}w^2dx}w^2-(\nabla\mathcal{G}_D)^T\mathcal{H}^2\mathbf{a}^0w^2=0,	
		\end{aligned}	
		\end{equation}	
		where
		$${\bf a}^0=\lim\limits_{\e\to0}{\bf a}^\e=\lim\limits_{\e\to0}(a_1^\e,a_2^\e)^T.$$	
		Thus
		\begin{equation}
		\label{8.30}
		\begin{aligned}
		\Phi_0
		=-\left(I-2\mathcal{G}_D\mathcal{H}\right)^{-1}(\nabla\mathcal{G}_D)^T\mathcal{H}^2{\bf a}^0w
		=-\mathcal{P}(\nabla \mathcal{G}_D)^T\mathcal{H}^2{\bf a}^0w,
		\end{aligned}	
		\end{equation}	
		where
		$$\mathcal{P}=(I-2\mathcal{G}_D\mathcal{H})^{-1}.$$
		Now we compare $\Phi_0$ with $\phi_\e^1$. By definition
		\begin{equation}
		\label{8.31}
		\phi_\e^1
		=-\e\sum_{j=1}^2\sum_{i=1}^2a_i^\e\nabla_{p_i^\e}\hat\xi_jw\left(y-\frac{p_j^\e}{\e}\right).	
		\end{equation}
		On the other hand,
		\begin{equation}
		\label{8.31-1}
		\phi_\e^\perp=\e\sum_{i=1}^2\tilde\phi_{\e,i}+h.o.t.
		=\e\sum_{i=1}^2\phi_i\left(y-\frac{p_i^\e}{\e}\right)
		+o(\e).
		\end{equation}	
		{The lemma is proved by} using \eqref{8.30} and comparing \eqref{8.31} and \eqref{8.31-1}.
	\end{proof}
	
	From Lemma \ref{le8.3} we have that
	\begin{equation}
	\label{8.32}
	(\psi_\e^\perp(p_1^\e),\psi_\e^\perp(p_2^\e))^T	
	=-2\e\mathcal{G}_D\mathcal{H}\mathcal{P}(\nabla\mathcal{G}_D)^T\mathcal{H}^2{\bf a}^0+o(\e)		
	\end{equation}
	and
	\begin{equation}
	\label{8.33}
	\begin{aligned}
	\psi_\e^\perp(p_i^\e+\e z)-\psi_\e^\perp(p_i^\e)=
	2\e^2z\sum_{j=1}^2\nabla_{p_i^\e}G_D(p_i^\e,p_j^\e)\hat\xi_j\frac{\int_{\R}w\phi_jdx}{\int_{\R}w^2dx}+\e\Pi_i(y)
	+o(\e^2),
	\end{aligned}	
	\end{equation}
	where $\Pi_i(y)$~ {is an even function in}~$y.$

	With the above three lemmas we {can now} derive the following results concerning the three terms $\mathbf{J}_1,\mathbf{J}_2,\mathbf{J}_3$ {defined in \eqref{eq:J-def}.}
	\begin{lemma}
		\label{le8.4}	
		{Let $\mathcal{G}_D$, $\mathcal{H}$, $\mathcal{Q}$, and ${\bf a}_\e$ be given by \eqref{2.matrix-0}, \eqref{2.h}, \eqref{2.q}, \eqref{8.14} respectively. Then}
		\begin{equation}
		\label{8.34}
		\begin{aligned}
		\mathbf{J}_1=c_1\e^2\mathcal{H}(\nabla^2\mathcal{G}_D+\mathcal{Q})\mathcal{H}^2{\bf a}_\e+o(\e^2),~\qquad ~	
		\mathbf{J}_2=2c_1\e^2\mathcal{H}\nabla \mathcal{G}_D\mathcal{H}\mathcal{P}(\nabla\mathcal{G}_D)^T\mathcal{H}^2{\bf a}_\e+o(\e^2),	
		\end{aligned}		
		\end{equation}
		and $\mathbf{J}_3=o(\e^2)$, where $c_1=\frac13\int_{\R}w^3dy$. 
	\end{lemma}
	
	\begin{proof}	
		The computation of $\mathbf{J}_1$ follows from Lemma \ref{le8.2}. In fact since $\oh'_\e=o(1)$
		\begin{equation}
		\label{8.35}
		\begin{aligned}
		J_{1,l}=&~\sum_{j=1}^2a_j^\e\int_{-\frac{1}{\e}}^{\frac{1}{\e}}\frac{\tw_{\e,l}^2}{\oh_\e^2}(\psi_{\e,j}-\e\oh'_\e\delta_{jl})\tw'_{\e,l}+h.o.t.\\
		=&~\sum_{j=1}^2a_j^\e\int_{-\frac{1}{\e}}^{\frac{1}{\e}}\frac{\tw_{\e,l}^2}{\oh_\e^2}\left([\psi_{\e,j}(y)-\e\oh'_\e(y)\delta_{jl}]-[\psi_{\e,j}(p_l^\e)-\e\oh'_{\e}(p_l^\e)\delta_{jl}]\right)\tw'_{\e,l}dy+o(\e^2)\\
		=~&-\e^2\hat\xi_l\int_{\R}(yw^2w'(y))dy\sum_{j=1}^2a_j^\e\left(\nabla_{p_l^\e}\nabla_{p_j^\e}G_D(p_l^\e,p_j^\e)+q_{lj}\delta_{lj}\right)\hat\xi_j^2+o(\e^2)\\
		=~&c_1\e^2\hat\xi_l\sum_{j=1}^2a_j^\e\left(\nabla_{p_l^\e}\nabla_{p_j^\e}G_D(p_l^\e,p_j^\e)+q_{lj}\delta_{lj}\right)\hat\xi_j^2+o(\e^2),		
		\end{aligned}	
		\end{equation}	
		which, by Lemma \ref{le8.1}, proves the first estimate of \eqref{8.34}. {The estimate for $\mathbf{J}_2$ follows from}
		\begin{equation}
		\label{8.36}
		\begin{aligned}
		{J}_{2,l}=~&\int_{-\frac{1}{\e}}^{\frac{1}{\e}}\frac{\tw_{\e,l}^2}{\oh_\e^2}\psi_\e^\perp\tw'_{\e,l}dy\\	
		=~&\int_{-\frac{1}{\e}}^{\frac{1}{\e}}\frac{\tw_{\e,l}^2}{\oh_{\e}^2}\psi_\e^\perp(p_l^\e)\tw'_{\e,l}dy+\int_{-\frac{1}{\e}}^{\frac{1}{\e}}\frac{\tw_{\e,l}^2}{\oh_{\e}^2}(\psi_\e^\perp(x)-\psi_\e^\perp(p_l^\e))\tw'_{\e,l}dy\\
		=~&\int_{-\frac{1}{\e}}^{\frac{1}{\e}}\frac{\tw_{\e,l}^2}{\oh_{\e}^2}(\psi_\e^\perp(x)-\psi_\e^\perp(p_l^\e))\tw'_{\e,l}dx+o(\e^2),
		\end{aligned}		
		\end{equation}
		{together with \eqref{8.30}, \eqref{8.32}-\eqref{8.33}}. The estimate on $J_{3,l}$ follows from Lemma \ref{le8.3}, the fact that {$\oh_\e(p_l^\e)=\hat\xi_l+O(\e^{2s-1})$ at $p_l^\e$} and the leading order of $\oh'_\e(p_l^\e+\e y)-\oh'_\e(p_l^\e)$ is an odd function of order $\e$.
	\end{proof}
	
	{We can now provide an estimate on the small eigenvalue. From Lemma \ref{le8.4} we have}
	\begin{equation*}
	\begin{split}
	\mathbf{J}_1+\mathbf{J}_2+\mathbf{J}_3
	& =c_1\e^2\mathcal{H}\left((\nabla^2\mathcal{G}_D+\mathcal{Q})\mathcal{H}^2+2\nabla\mathcal{G}_D\mathcal{H}\mathcal{P}(\nabla\mathcal{G}_D)^T\mathcal{H}^2\right){\bf a}_\e+o(\e^2)
	\\
	& =c_1\e^2\mathcal{H}^2\mathcal{M}(\p^\e){\bf a}_\e+o(\e^2),
	\end{split}
	\end{equation*}
	{and therefore by combining} \eqref{8.15} and \eqref{8.16}, we obtain
	\begin{equation}
	\label{8.39}
	c_1\e^2\mathcal{H}^2\mathcal{M}(\p^\e){\bf a}_\e+o(\e^2)=-\la_\e\mathcal{H}^2{\bf a}_\e\int_{\R}(w'(y))^2dy(1+o(1)).
	\end{equation}
	From this equation we see that
	$$\la_\e=-\e^2c_2\lambda_{\mathcal{M}(\p^0)}(1+o(1)),$$	
	where $c_2$ is a positive constant and $\lambda_{\mathcal{M}(\p^0)}$ is the non-zero left eigenvalue of $\mathcal{M}(\p^0)$ given in \eqref{2.me}. In Appendix \ref{app:greens-func} we derive a quickly converging series expression for the Green's function for which we can interchanging summation and differentiation to calculate its second derivatives. Numerical calculations then indicate that $\partial_x^2G_D(x,0)>0$ at $x=1$ and using \eqref{2.me} we therefore deduce that the non-zero left eigenvalue of $\mathcal{M}(\p^0)$ is positive. The small $\la_\varepsilon$ is therefore negative (stable) so that the two spike pattern is linearly stable with respect to the small eigenvalues. In particular, linear stability is determined solely by the eigenvalues of $\mathcal{B}$ and the proof of Theorem \ref{th1.stability} is therefore complete.

\section{Conclusion and Open Problems}\label{sec:conclusion}

In this paper we have proven the existence and rigorously analyzed the stability of both symmetric and asymmetric two spike equilibrium solutions of the fractional one-dimensional Gierer-Meinhardt system \eqref{1.fgm} with periodic boundary conditions. In addition, by using a combination of formal asymptotic and numerical methods we have calculated asymptotic approximations for $N$-spike quasi-equilibrium solutions and derived a system of ODEs governing their slow dynamics on an $\mathcal{O}(\varepsilon^{-2})$ timescale as well as a system of NLEPs governing their linear stability on an $\mathcal{O}(1)$ timescale. Our findings indicate that a single spike solution may be destabilized or stabilized with respect to oscillatory instabilities by decreasing the fractional exponents for the activator, $s_1$, or inhibitor, $s_2$, respectively. On the other hand we found that decreasing the fractional exponent for the inhibitor, $s_2$, has a stabilizing effect on the stability of symmetric two-spike solutions with respect to competition instabilities. Finally we determined that asymmetric two-spike solutions are always linearly unstable with respect to competition instabilities. In all one- and two-spike cases we found that the equilibrium spike patterns are linearly stable with respect to the slow dynamics and that this is a consequence of the choice of periodic boundary conditions.

We conclude this section with an outline of open problems and directions for future research. The first open problem is to prove the existence and to provide a complete classification of all $N$-spike equilibrium solutions to the fractional one-dimensional Gierer-Meinhardt model. In particular a key question is whether, as in the classical Gierer-Meinhardt model \cite{wei_2007_existence}, asymmetric $N$-spike solutions are generated by sequences of spikes of two types. In addition it would be interesting to extend our results to the fractional Gierer-Meinhardt model with different boundary conditions (e.g.\@ Neumann or Dirichlet) for which our analysis needs to be extended in order to provide regularity estimates at the boundaries $x=\pm1$. Another interesting direction for future research is to investigate the behaviour of solutions to the fractional GM model in the $D\ll 1$ regimes for which the classical GM model is known to exhibit distinct behaviour such as spike splitting and clustering. Finally a detailed analysis, either rigorous or formal, of localized solutions for different reaction-kinetics as well as in higher-dimensional domains would be a fruitful direction of future research.

\addcontentsline{toc}{section}{References}
\bibliographystyle{abbrv}
\bibliography{biblio}

\appendix

\section{The nonlocal eigenavlue problem}\label{app:nonlocal}
In this section, we prove Theorem \ref{th3.stability}. We consider the eigenvalue problem
\begin{equation}
\label{a.1}
\begin{aligned}
(-\Delta)^s\phi+\phi-2w\phi+\gamma\frac{\int_{\R}w\phi dx}{\int_{\R}w^2dx}+\alpha\phi=0,\quad \phi\in H^{2s}(\R).
\end{aligned}
\end{equation}
Our aim is to show that the above eigenvalue problem has an eigenvalue with real part when $\gamma\in(0,1)$ and the real part of the eigenvalue is always negative if $\gamma>1$ and $s>\frac14$.

Before we give the proof of Theorem \ref{th3.stability}, we first present the following result.
\begin{proposition}[\cite{frank_2013_uniqueness}]
	\label{pra.1}
	The  eigenvalue problem
	\begin{equation}
	\label{a.eigen}
	(-\Delta)^s\phi+\phi-2w\phi+\mu\phi=0~\mathrm{in}~\R,\quad\phi\in H^{2s}(\R)
	\end{equation}
	admits the following set of eigenvalues:
	\begin{equation}
	\label{a.eigen1}
	\mu_1>0,\quad \mu_2=0,\quad \mu_{3}<0,\quad \cdots.
	\end{equation}
	Moreover, the eigenfunction corresponding to $\mu_1$ is radial and of costant sign.
\end{proposition}

\begin{proof}[Proof of Theorem \ref{th3.stability}-(1).]
	The original problem is equivalent to finding a positive zero root  of the function $\mathcal{F}(\alpha)$ defined by
	\begin{equation*}
	\mathcal{F}(\alpha)=\int_{\R}w^2dx+\gamma\int_{\R}w(L_0+\alpha)^{-1}w^2dx,
	\end{equation*}
	where
	$$L_0\phi=(-\Delta)^s\phi+\phi-2w\phi.$$
	By the above proposition, $L_0$ has a unique eigenvalue $\mu_1>0$ with an eigenfunction of constant sign. We now consider $\mathcal{F}(\alpha)$ in the interval $(0,\mu_1).$ Since $L_0^{-1}w^2=-w$, we deduce that
	\begin{equation}
	\label{a.small}
	\mathcal{F}(0)=(1-\gamma)\int_{\R}w^2dx>0,
	\end{equation}
	provided $\gamma<1$. Next, as $\alpha\to\mu_1^{-}$, we have that
	\begin{equation}
	\label{a.small1}
	\int_{\R}w(L_0+\alpha)^{-1}w^2dx\to -\infty.
	\end{equation}
	Hence, we get from \eqref{a.small1} that
	\begin{equation}
	\label{a.small2}	
	\mathcal{f}(\alpha)\to-\infty\quad \mathrm{as}\quad \alpha\to\mu_1^{-1},	
	\end{equation}
	when $\gamma\in(0,1)$. By \eqref{a.small}, \eqref{a.small2} and the continuity of $\mathcal{f}(\alpha)$, we can find a $\alpha_0\in(0,\mu_1)$ such that $f(\alpha_0)=0$ whenever $\gamma\in(0,1).$
\end{proof}

Next, we shall study \eqref{a.1} when $\gamma>1$. We shall prove that the real part of the eigenvalue is negative in any case. To this end, we introduce some notation and make some preparations. Set
\begin{equation}
\label{a.2}
\mathcal{L}\phi:=L_0\phi+\gamma\frac{\int_{\mathbb{R}}w\phi dx}{\int_{\mathbb{R}}w^2dx}w^2,\quad  \phi\in H^{2s}(\R).
\end{equation}

According to the definition of $L$, we can easily see that $L$ is not self-adjoint. Let
$$X_0:=\mathrm{kernel}(L_0)=\mathrm{Span}\left\{\frac{\partial w}{\partial x}\right\}.$$
Then
\begin{equation}
\label{a.3}
L_0w=-w^2,\quad L_0\left(w+\frac{1}{2s}x\cdot\nabla w\right)=-w.
\end{equation}
Hence
\begin{equation}
\label{a.4}
\int_{\R}(L_0^{-1}w)wdx=\int_{\R}\left(-\frac{1}{2s}x\cdot\nabla w-w\right)wdx=\frac{1-4s}{4s}\int_{\R}w^2dx,
\end{equation}
and
\begin{equation}
\label{a.5}
\int_{\R}(L_0^{-1}w)w^pdx=-\int_{\R}L_0^{-1}wL_0wdx=-\int_{\R}w^2dx.
\end{equation}

Before we give the proof of Theorem \ref{th3.stability}. We present the following important lemma.

\begin{lemma}
	\label{lea.2}
	Let $L_1$ be an operator defined by
	\begin{equation}
	\label{a.6}
	L_1\phi=L_0\phi+\frac{\int_{\R}w\phi dx}{\int_{\R}w^2dx}w^2+
	\frac{\int_{\R}w^2\phi dx}{\int_{\R}w^2dx}w-\frac{\int_{\R}w^3dx\int_{\R}w\phi dx}{\left(\int_{\R}w^2dx\right)^2}w.
	\end{equation}
	Then we have
	\begin{enumerate}
		\item [(1)] $L_1$ is self-adjoint and the kernel of $L_1$ (denoted by $X_1$) is $\mathrm{Span}\left\{w,~\frac{\partial w}{\partial x}\right\}$.
		
		\item [(2)] There exists a positive constant $a_1>0$ such that
		\begin{equation}
		\label{a.7}
		\begin{aligned}
		L_1(\phi,\phi):=~&\int_{\R}\left(|(-\Delta)^{\frac{s}{2}}\phi|^2+\phi^2-2w\phi^2\right)dx
		+2\frac{\int_{\R}w\phi dx\int_{\R}w^2\phi dx}{\int_{\R}w^2dx}-\frac{\int_{\R}w^3dx\left(\int_{\R}w\phi dx\right)^2}{\left(\int_{\R}w^2dx\right)^2}\\
		\geq~&a_1 d_{L^2(\R)}^2(\phi,X_1),
		\end{aligned}
		\end{equation}
		for all $\phi\in H^{2s}(\R)$, where $d_{L^2(\R)}$ means the distance in $L^2$-norm.
	\end{enumerate}
\end{lemma}

\begin{proof}
	By \eqref{a.7}, $L_1$ is self-adjoint. It is easy to see that $w,\frac{\partial w}{\partial y}\in\mathrm{Kernel}(L_1)$. On the other hand, if $\phi\in\mathrm{Kernel}(L_1)$, then by Proposition \ref{pra.1}
	\begin{equation*}
	L_0\phi=-c_1(\phi)w-c_2(\phi)w^2=c_1(\phi)L_0(w+\frac{1}{2s}x\cdot\nabla w)+c_2(\phi)L_0(w),
	\end{equation*}
	where
	\begin{equation*}
	c_1(\phi)=\frac{\int_{\R}w^2\phi dx}{\int_{\R}w^2dx}-\frac{\int_{\R}w^3dx\int_{\R}w\phi dx}{\left(\int_{\R}w^2dx\right)^2},\quad
	c_2(\phi)=\frac{\int_{\R}w\phi dx}{\int_{\R}w^2dx}.
	\end{equation*}
	Hence
	\begin{equation}
	\label{a.8}
	\phi-c_1(\phi)(w+\frac{1}{2s}x\cdot\nabla w)-c_2(\phi)w \in \mathrm{kernel}(L_0).
	\end{equation}
	Note that
	\begin{align*}
	c_1(\phi)=~&c_1(\phi)\frac{\int_{\R}w^2\left(w+\frac{1}{2s}
		x\cdot\nabla w\right)dx}{\int_{\R}w^2dx}
	-c_1(\phi)\frac{\int_{\R}w^3dx\int_{\R}w(w+\frac{1}{2s}x\cdot\nabla w)dx}{\left(\int_{\R}w^2dx\right)^2}\\
	=~&c_1(\phi)-c_1(\phi)(1-\frac{1}{4s})\frac{\int_{\R}w^3dx}{\int_{\R}w^2dx}
	\end{align*}
	by \eqref{a.4} and \eqref{a.5}. This implies that $c_1(\phi)=0$ for $s>\frac14$. By \eqref{a.8} and Proposition \ref{pra.1}, we prove the first conclusion.
	\medskip
	
	It remains to prove (2). Suppose it is not true. Then by the first conclusion there exists $(\alpha,\phi)$ such that
	$(i)$ $\alpha$ is real and positive, $(ii)$ $\phi\perp w,~\phi\perp\frac{\partial w}{\partial x}$ and $(iii)$ $L_1(\phi)+\alpha\phi=0.$
	
	We shall show the above conclusion is not possible. From $(ii)$ and $(iii)$ we have
	\begin{equation}
	\label{a.9}
	(L_0+\alpha)\phi+\frac{\int_{\R}w^2\phi dx}{\int_{\R}w^2dx}w=0.
	\end{equation}
	First we claim that $\int_{\R}w^2\phi\neq 0$. In fact if $\int_{\R}w^2\phi=0$, then $-\alpha<0$ is an eigenvalue of $L_0$. By Proposition \ref{pra.1}, $-\alpha=\mu_1$ and $\phi$ has costant sign. This contradicts with the fact that $\phi\perp w$. Therefore $-\alpha\neq\mu_1,0$ and hence $L_0+\alpha$ is invertible in $X_0^\perp$. So \eqref{a.9} implies
	\begin{equation*}
	\phi=-\frac{\int_{\R}w^2\phi dx}{\int_{\R}w^2}(L_0+\alpha)^{-1}w.
	\end{equation*}
	Thus
	\begin{equation*}
	\int_{\R}w^2\phi dx=-\frac{\int_{\R}w^2\phi dx}{\int_{\R}w^2dx}\int_{\R}((L_0+\alpha)^{-1}w)w^2dx,
	\end{equation*}
	which implies
	\begin{equation*}
	\int_{\R}w^2dx=-\int_{\R}((L_0+\alpha)^{-1}w)w^2dx
	=\int_{\R}((L_0+\alpha)^{-1}w)((L_0+\alpha)w-\alpha w)dx,
	\end{equation*}
	hence
	\begin{equation}
	\label{a.10}
	\int_{\R}((L_0+\alpha)^{-1}w)wdx=0.
	\end{equation}
	Let $h_1(\alpha)=\int_{\R}((L_0+\alpha)^{-1}w)wdx$, then $$h_1(0)=\int_{\R}(L_0^{-1}w)wdx=-\int_{\R}(w+\frac{1}{2s}x\cdot\nabla w)w=(\frac{1}{4s}-1)\int_{\R}w^2dx<0,$$
	due to $s<\frac14$. Moreover,
	$$h_1'(\alpha)=-\int_{\R}(L_0+\alpha)^{-2}ww=-\int_{\R}((L_0+\alpha)^{-1}w)^2dx<0.$$
	This shows that $h_1(\alpha)<0$ for all $\alpha\in(0,\mu_1)$. Clearly, $h_1(\alpha)>0$ for all $\alpha\in(\mu_1,\infty)$ since $\lim\limits_{\alpha\to+\infty}h_1(\alpha)=0$. This is a contradiction to \eqref{a.10} and we finish the proof.
\end{proof}

\begin{proof}[Proof of Theorem \ref{th3.stability}-(2)-(3).]
	We now finish the proof of Theorem \ref{th3.stability}-(2) and (3). First, we prove (2). Let
	$\alpha_0=\alpha_R+i\alpha_I$ and $\phi=\phi_R+i\phi_I$. Since $\alpha_0\neq0$, we can choose $\phi\perp\mathrm{kernel}(L_0)$. Then we can obtain two equations
	\begin{equation}
	\label{a.11}
	\begin{cases}
	L_0\phi_R+\gamma\frac{\int_{\R}w\phi_Rdx}{\int_{\R}w^2dx}w^2=-\alpha_R\phi_R+\alpha_I\phi_I,\\
	L_0\phi_I+\gamma\frac{\int_{\R}w\phi_Idx}{\int_{\R}w^2dx}w^2=-\alpha_R\phi_I-\alpha_I\phi_R,
	\end{cases}
	\end{equation}
	Multiplying the first equation of \eqref{a.11} by $\phi_R$ and the second one of \eqref{a.11} by $\phi_I$ and adding them together, we obtain
	\begin{equation}
	\label{a.12}
	\begin{aligned}
	-\alpha_R\int_{\R}(\phi_R^2+\phi_I^2)dx=~&L_1(\phi_R,\phi_R)+L_1(\phi_I,\phi_I)
	+\dfrac{\int_{\R}w^3dx}{\left(\int_{\R}w^2dx\right)^2}
	\left[\left(\int_{\R}w\phi_Rdx\right)^2+\left(\int_{\R}w\phi_Idx\right)^2\right]\\
	&+(\gamma-2)\frac{\int_{\R}w\phi_Rdx\int_{\R}w^2\phi_Rdx+\int_{\R}w\phi_Idx\int_{\R}w^2\phi_Idx}{\int_{\R}w^2dx}
	\end{aligned}
	\end{equation}
	Multiplying both equations of \eqref{a.11} by $w$ and adding together, we get
	\begin{equation}
	\label{a.13}
	\begin{aligned}
	\int_{\R}w^2\phi_Rdx-\gamma\frac{\int_{\R}w\phi_Rdx}{\int_{\R}w^2dx}\int_{\R}w^3dx=\alpha_R\int_{\R}w\phi_Rdx-\alpha_I\int_{\R}w\phi_Idx,\\
	\int_{\R}w^2\phi_Idx-\gamma\frac{\int_{\R}w\phi_Idx}{\int_{\R}w^2dx}\int_{\R}w^3dx=\alpha_R\int_{\R}w\phi_Idx+\alpha_I\int_{\R}w\phi_Rdx.
	\end{aligned}
	\end{equation}
	We multiply the first equation of \eqref{a.13} by $\int_{\R}w\phi_Rdx$ and the second one of \eqref{a.13} by $\int_{\R}w\phi_Idx$ and add them together, we obtain
	\begin{equation}
	\label{a.14}
	\int_{\R}w\phi_Rdx\int_{\R}w^2\phi_Rdx+\int_{\R}w\phi_Idx\int_{\R}w^2\phi_Idx
	=\left(\alpha_R+\gamma\frac{\int_{\R}w^3dx}{\int_{\R}w^2dx}\right)\left((\int_{\R}w\phi_Rdx)^2+(\int_{\R}w\phi_Idx)^2\right).
	\end{equation}
	Therefore, we have
	\begin{equation}
	\label{a.15}
	\begin{aligned}
	-\alpha_R\int_{\R}(\phi_R^2+\phi_I^2)dx
	=~&L_1(\phi_R,\phi_R)+L_1(\phi_I,\phi_I)+\frac{\int_{\R}w^3dx}{(\int_{\R}w^2dx)^2}\left[(\int_{\R}w\phi_Rdx)^2+(\int_{\R}w\phi_Idx)^2\right]\\
	&+(\gamma-2)\left(\alpha_R+\gamma\frac{\int_{\R}w^3dx}{\int_{\R}w^2dx}\right)\frac{(\int_{\R}w\phi_Rdx)^2+(\int_{\R}w\phi_Idx)^2dx}{\int_{\R}w^2dx}.
	\end{aligned}
	\end{equation}
	Set
	$$\phi_R=c_Rw+\phi_R^\perp,~\phi_R^\perp\perp X_1,\quad
	\phi_I=c_Iw+\phi_I^\perp,~\phi_I^\perp\perp X_1.$$
	Then
	$$\int_{\R}w\phi_Rdx=c_R\int_{\R}w^2dx,\quad \int_{\R}w\phi_Idx=c_I\int_{\R}w^2dx,$$
	and
	$$d_{L^2(\R)}^2(\phi_R,X_1)=\|\phi_R^\perp\|_{L^2}^2,\quad
	d_{L^2(\R)}^2(\phi_I,X_1)=\|\phi_I^\perp\|_{L^2}^2.$$
	By some simple computations we have
	\begin{equation*}
	L_1(\phi_R,\phi_R)+L_1(\phi_I,\phi_I)+\alpha_R(\gamma-1)
	(c_R^2+c_I^2)\int_{\R}w^2dx+(c_R^2+c_I^2)(\gamma-1)
	^3\int_{\R}w^3dx+\alpha_R(\|\phi_R^\perp\|_{L^2}^2+\|\phi_I^\perp\|_{L^2}^2)=0.
	\end{equation*}
	By Lemma \ref{lea.2},
	\begin{equation*}
	\alpha_R(\gamma-1)(c_R^2+c_I^2)\int_{\R}w^2dx+(\gamma-1)^2(c_R^2+c_I^2)\int_{\R}w^3dx+\alpha_R(\|\phi_R^\perp\|_{L^2}^2+\|\phi_I^\perp\|_{L^2}^2)\leq0.
	\end{equation*}
	Since $\gamma<1$, we have $\alpha_R<0$, which proves Theorem \ref{th3.stability}-(2).
	
	It remains to prove the last conclusion. Since $\phi$ satisfies
	\begin{equation}
	\label{a.16}
	\s\phi+\phi-2w\phi+\gamma\frac{\int_{\R}w\phi dx}{\int_{\R}w^2dx}w^2=0.
	\end{equation}
	Then $L_0\phi=-c_3(\phi)w$,	where  $c_3(\phi)=\gamma\frac{\int_{\R}w\phi dx}{\int_{\R}w^2dx}.$
	Hence  $\phi-c_3(\phi)w\in \mathrm{Kernel}{(L_0)}.$
	Thus
	\begin{equation}
	\label{a.17}	
	c_3(\phi)\gamma=\gamma\frac{\int_{\R}w\phi dx}{\int_{\R}w^2dx}=c_3(\phi).
	\end{equation}
	So if $\gamma\neq1$, we get $c_3(\phi)=0$. Then $\phi\in\mathrm{Kernel}{(L_0)}$ and we complete the proof.
\end{proof}
\medskip

\section{Overview of Numerical Calculations}\label{app:numerical}

In this section we briefly outline the numerical calculation of solutions to the core problem \eqref{eq:core-problem} and the time-dependent fractional GM system with periodic boundary conditions \eqref{1.fgm}. In both cases we use the finite difference-quadrature discretization for the fractional Laplacian with piecewise linear interpolants developed by Huang and Oberman \cite{huang_2013}. When discretizing \eqref{eq:core-problem} we approximate the fractional Laplacian on a truncated domain using the far-field behaviour presented in Proposition \ref{pr3.1} to capture the nonlocal behaviour outside the truncated domain. On the other hand, when spatially discretizing \eqref{1.fgm} we use the spatial periodicity of the system to simplify the expression for the discrete fractional Laplacian. Time stepping of the spatially discretized system is then performed using a second-order semi-implicit backwards difference scheme \cite{ruuth_1995}. In the remainder of this section we provide additional details for both of these cases.


First we consider the numerical calculation of solutions to the core problem \eqref{eq:core-problem}. Since the domain for \eqref{eq:core-problem} is $-\infty<y<\infty$ we need to both truncate and then discretize the truncated domain to obtain a numerical calculation. Outside of the truncated domain we use the far-field behaviour from Proposition \ref{pr3.1} to impose a Dirichlet boundary condition. Specifically, letting $L>0$ we approximate solutions to \eqref{eq:core-problem} by solving the truncated problem
\begin{equation*}
(-\Delta)^s U + U - U^2 = 0,\quad |y|<L,\qquad U(y) = U(L)(L/y)^{1+2s},\quad |y|\geq L,
\end{equation*}
where we have replaced $\mathfrak{b}_s$ with $U(L) L^{1+2s}$ since we do not yet know the value of $\mathfrak{b}_s$. To account for the nonlocal contributions outside of the truncated domain we discretize a computational domain that extends beyond the truncated domain. Specifically we discretize the computational domain $-2L\leq y\leq 2L$ by letting $y_i = ih$ for $i=-2N,...,2N$ where $h=1/N$. Seeking symmetric solutions we impose $U_i = U_{|i|}$ for all $i=-2N,...,2N$ which reduces the unknown values to $U_0,...,U_N$. Note in addition that $U_i = (L/y_{|i|})^{1+2s} U_N $ for all $N<|i|\leq 2N$. The fractional Laplacian can then be approximated by (see \S5 in \cite{huang_2013})
\begin{equation}\label{eq:core-problem-discretized}
(-\Delta)^sU(y_i) \approx (-\Delta_h)^sU_i = \sum_{j=-2N}^{2N}(U_i-U_{|i-j|})w_{j} +  C^{II}U_i - C_{i}^{III}U_N,\qquad i=0,...,N
\end{equation}
where the first term accounts for integration inside of the truncated domain and
\begin{equation}\label{eq:discretized_weights}
w_j = \frac{C_s}{2s(2s-1)h^{2s}}\begin{cases} 2^{1-2s}-2+(1-s)^{-1}s , & j=\pm 1, \\ |j+1|^{1-2s} - 2|j|^{1-2s} + |j-1|^{1-2s}, & \text{otherwise},\end{cases}\quad (j=\pm 1, \pm 2, ...),
\end{equation}
where we note that the value of $w_0$ is never needed in the discretization. The remaining two terms $C^{II}$ and $C_i^{III}$ account for contributions outside of the computational domain and are respectively given by
\begin{gather*}
C^{II} = \frac{C_s}{s(2L)^{2s}}, \\
C_i^{III} = \frac{C_s L^{2s+1}}{(4s+1)(2L)^{4s+1}}\biggl( \prescript{}{2}{F}_1\bigl(2s+1,4s+1;4s+2,\tfrac{y_i}{2L}\bigr) + \prescript{}{2}{F}_1\bigl(2s+1,4s+1;4s+2,-\tfrac{y_i}{2L}\bigr) \biggr),
\end{gather*}
where $\prescript{}{2}{F}_1$ is the Gaussian hypergeometric function.

With the above discretization it is then possible to approximate solutions to \eqref{eq:core-problem} by solving the nonlinear algebraic system \eqref{eq:core-problem-discretized} for the $N+1$ unknowns $U_0,...,U_N$. To numerically solve this nonlinear system we use the \texttt{fsolve} function in the Python 3.6.8 SciPy library. Our initial guess for the nonlinear solver is obtained by numerical continuation in $s$ starting with $s=1/2$ for which the exact solution $w_s=2/(1+y^2)$ is known. In this way we may numerically calculate the core solution for an arbitrary value of $s$ and in Figure \ref{fig:core_problem-sol} we plot the resulting core solutions for select values of $s$ where we have used $N=2000$ for the spatial discretization. From these solutions we may also extract the value of the far-field decay coefficient $\mathfrak{b}_s$ and this is plotted in Figure \ref{fig:core_problem-constant}. We conclude by remarking that no nontrivial solution to the core problem \eqref{eq:core-problem} exists for $s\leq 1/6$ (see for example \cite{frank_2013_uniqueness}) and our numerical computations failed to yield solutions for $s\approx 0.2$ and below because of this.

\begin{figure}
	\centering
	\begin{subfigure}{0.5\textwidth}
		\centering
		\includegraphics[scale=0.675]{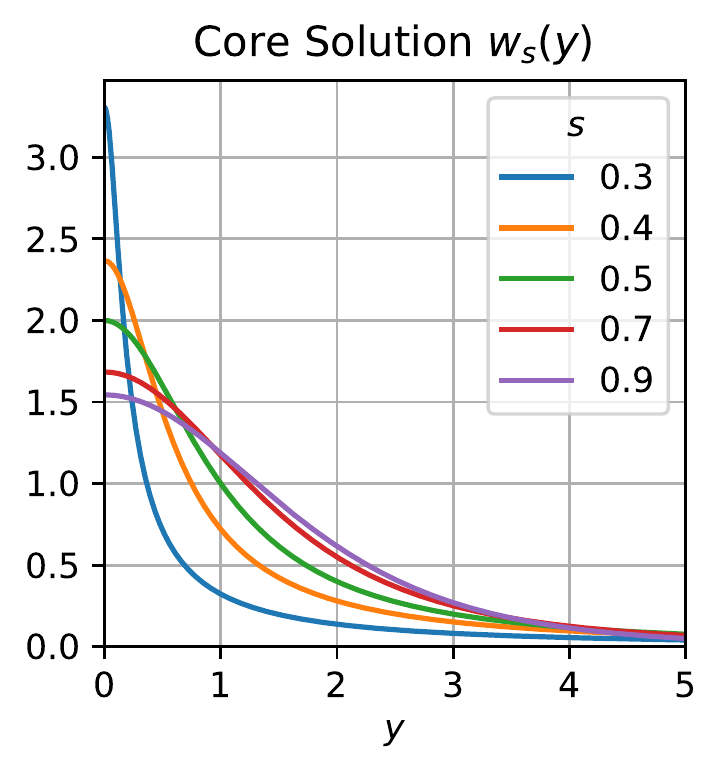}
		\caption{}\label{fig:core_problem-sol}
	\end{subfigure}%
	\begin{subfigure}{0.5\textwidth}
		\centering
		\includegraphics[scale=0.675]{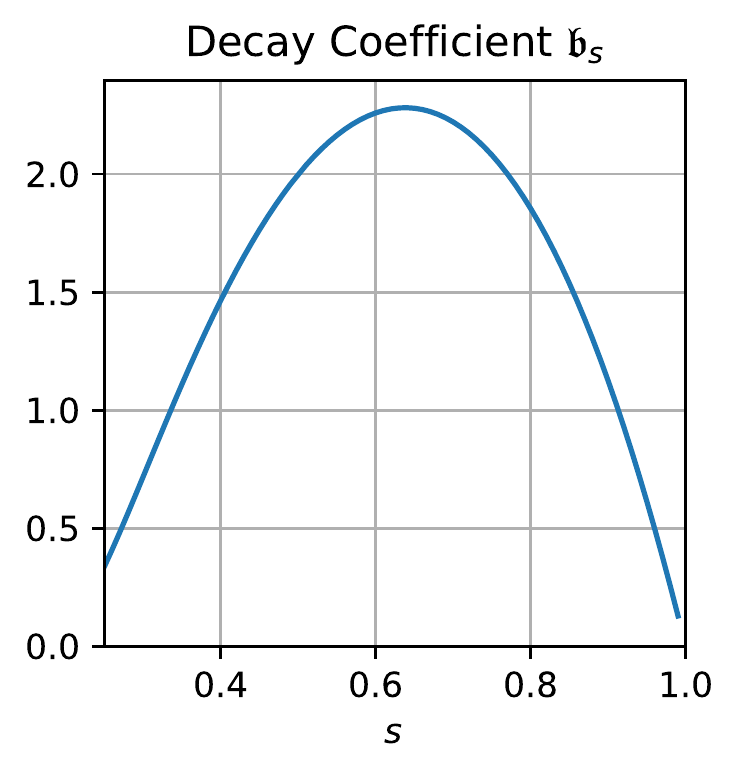}
		\caption{}\label{fig:core_problem-constant}
	\end{subfigure}
	\caption{(A) Sample plots of numerically computed solutions to the core problem \eqref{eq:core-problem}. (B) Far-field decay coefficient $\mathfrak{b}_s$ in the core problem \eqref{eq:core-problem}.}\label{fig:core_problem}
\end{figure}

Turning now to the numerical solution of \eqref{1.fgm} we discretize the interval $-1<x<1$ into $N$ uniformly distributed points given by $x_i = -1 + 2ih$ for $i=0,...,N-1$ where $h=1/N$. Assuming that $\phi(x)$ is a $2-$periodic function on $-1<x<1$ and letting $\phi_i\equiv \phi(x_i)$ for each $i=0,...,N-1$ we calculate (see equation (FLh) in \cite{huang_2013})
\begin{equation}\label{eq:discretized_gm_1}
(-\Delta)^s\phi(x_i)\approx(-\Delta_h)^s\phi_i = \sum_{j=-\infty}^\infty (\phi_i - \phi_{i-j})w_j = \sum_{j=0}^{N-1} W_{i-j} (\phi_i - \phi_j),
\end{equation}
where the final equality follows from the periodicity of $\phi$ and where
$$
W_\sigma \equiv w_\sigma + \sum_{k=1}^\infty(w_{\sigma+Nk} + w_{\sigma-Nk}),
$$
with each weight $w_i$ ($i\in\mathbb{Z}$) being given by \eqref{eq:discretized_weights}. In our numerical calculations we truncate the sum after $500$ terms. From \eqref{eq:discretized_gm_1} it is then straightforward to deduce the entries of the matrix $(-\Delta_h)^s$ which we remark is dense in contrast to the tridiagonal matrix obtained by applying a finite-difference approximation to the one-dimensional Laplacian. With this spatial discretization we can then approximate \eqref{1.fgm} with the $2N$-dimensional system of ODEs
\begin{equation}\label{eq:discretized_gm_2}
\frac{d\Phi}{dt} + \mathcal{A}\Phi + \mathcal{N}(\Phi) = 0,
\end{equation}
where $\Phi(t) = (u_0(t),...,u_{N-1}(t),v_0(t),...,v_{N-1}(t))^T$, $\mathcal{A}=\text{diag}(\varepsilon^{2s_1}(-\Delta_h)^{s_1}, \tau^{-1}D(-\Delta_h)^{s_2})$, and $\mathcal{N}(\Phi)$ is the $2N$-dimensional array that accounts for the nonlinearities in \eqref{1.fgm}. To integrate \eqref{eq:discretized_gm_2} we employ a second-order semi-implicit backwards difference scheme (2-SBDF) \cite{ruuth_1995} that uses second-order backward difference time-stepping for the fractional Laplace term and explicit (forward) time-stepping for the nonlinear term. Specifically, given a time-step size $\Delta t>0$ and denoting by $\Phi_n = \Phi(t_n)$ where $t_n=n\Delta t$ the 2-SBDF scheme becomes
\begin{equation}
(3\mathcal{I}-2\Delta t {A})\Phi_{n+1} = 4\Phi_n - \Phi_{n-1} + 4\Delta t \mathcal{N}(\Phi_n) - 2\Delta t \mathcal{N}(\Phi_{n-1}),
\end{equation}
where $\mathcal{I}$ is the $2N\times 2N$ identity matrix. Given an initial condition $\Phi_0$ (based on the asymptotic approximations of \S\ref{sec:formal-results}) we also need $\Phi_1$ to initiate time-stepping with 2-SBDF. We calculate $\Phi_1$ by using a first-order semi-implicit backwards difference scheme (1-SBDF) \cite{ruuth_1995} given by
\begin{equation}
(\mathcal{I}-\Delta t \mathcal{A})\Phi_{n+1} = \Phi_n + \Delta t \mathcal{N}(\Phi_n),
\end{equation}
with which we perform five time steps with a step size that is one-fifth that used in our main 2-SBDF scheme. Throughout the numerical simulations of \eqref{1.fgm} in \S\ref{subsec:example-1} and \S\ref{subsec:example-2} we used a mesh consisting of $N=2000$ points and a time-step size of $\Delta t = 0.001$.

\section{A Rapidly Converging Series for the Fractional Green's Function}\label{app:greens-func}

In this section we provide a quickly converging series expansion of the Green's function $G_D(x,z)$ satisfying \eqref{3.green}. In particular, by adding and subtracting appropriate multiples of $|x-z|^{2s-1}$ and $|x-z|^{4s-1}$ as outlined below we obtain the series expansion
\begin{equation}
\begin{split}\label{eq:greens-fast-series}
G_D(x,z) & =  \mathfrak{a}_s\bigl(|x-z|^{2s-1}-\tfrac{1}{2s}\bigr) - \mathfrak{b}_s\bigl(|x-z|^{4s-1}-\tfrac{1}{4s}\bigr) - \tfrac{1}{2}\bigl((2s-1)\mathfrak{a}_s - (4s-1)\mathfrak{b}_s\bigr)\bigl(|x-z|^2-\tfrac{1}{3}\bigr) \\
& + \frac{1}{2} + \frac{1}{D^3}\sum_{n=1}^\infty\biggl(1 + \frac{1}{D(n\pi)^{2s}}\biggr)^{-1}\frac{\cos n\pi |x-z|}{(n\pi)^{6s}} + 2\sum_{n=1}^\infty \biggl(\frac{\mathfrak{b}_sb_n}{(n\pi)^{4s}} - \frac{\mathfrak{a}_sa_n}{(n\pi)^{2s}}\biggr)\cos n\pi |x-z|,
\end{split}
\end{equation}
where
\begin{equation}
\mathfrak{a}_s = -\frac{2}{\pi D}s\Gamma(-2s)\sin(\pi s),\quad \mathfrak{b}_s\equiv -\frac{4}{\pi D^2}s\Gamma(-4s)\sin(2\pi s),
\end{equation}
and
\begin{gather}\label{eq:an-bn-def}
a_n = (2s-1)(2s-2)\int_{n\pi}^\infty x^{2s-3}\cos xdx, \\
b_n = -(4s-1)(4s-2)(4s-3)\biggl((-1)^n(n\pi)^{4s-4} + (4s-4)\int_{n\pi}^\infty x^{4s-5}\cos x dx\biggr).
\end{gather}
The key reason for considering this expansion is that the coefficients of $\cos n\pi|x-z|$ converge to zero sufficiently fast to allow the order of summation and second-differentiation to be interchanged. In particular using \eqref{eq:greens-fast-series} we can numerically calculate that $\partial_x^2G_D(x,0)$ is strictly positive at $x=1$.

To derive \eqref{eq:greens-fast-series} we use integration by parts to calculate the coefficients in the Fourier series
\begin{equation}
|x|^{\beta-1} = \frac{1}{\beta} + 2\sum_{n=1}^\infty \frac{c_{n,\beta}}{(\pi n)^\beta}\cos n\pi x,\qquad c_{n,\beta} = \int_0^{n\pi}x^{\beta-1}\cos x dx,
\end{equation}
where $\beta=2s\in(1,2)$ or $\beta=4s\in(2,4)$. Specifically we calculate
\begin{align*}
c_{n,2s} & = -(2s-1)\int_0^{n\pi} x^{2s-2}\sin x dx  = -(2s-1)\int_0^\infty x^{2s-2}\sin xdx + (2s-1)\int_{n\pi}^\infty x^{2s-2}\sin x dx \\
& = -(2s-1)\int_0^\infty x^{2s-2}\sin xdx + (-1)^n (2s-1)(n\pi)^{2s-2} + a_n,
\end{align*}
for $\beta=2s$ and
\begin{align*}
c_{n,4s} = & (4s-1)(4s-2)(4s-3)\int_0^\infty x^{4s-4}\sin x dx + (-1)^{n}(4s-1)(n\pi)^{4s-2} + b_n,
\end{align*}
for $\beta=4s$ and where $a_n$ and $b_n$ are defined by \eqref{eq:an-bn-def}. The definite integrals appearing in $c_{n,2s}$ and $c_{n,4s}$ can then be written in terms of $\mathfrak{a}_s$ and $\mathfrak{b}_s$ respectively by using the integral representation of the Gamma function $\int_0^\infty x^{z-1}\sin x dx = \Gamma(z)\sin\bigl(\tfrac{\pi z}{2}\bigr)$ for $-1<\Re(z)<1$ together with the reflection formula $z\Gamma(z)\Gamma(-z) = -\pi/\sin\pi z$ for all $z\notin\mathbb{Z}$ (see equations 5.9.7 and 5.5.3 in \cite{NIST:DLMF} respectively).

\section{Derivation of the Slow Dynamics}\label{app:slow-dynamics}

In this appendix we outline the derivation of the system of ODEs \eqref{eq:slow-dynamics} governing the slow dynamics of the multi-spike quasi-equilibrium solutions considered in \S\ref{subsec:formal-equilibrium}. Letting $x=x_i+\varepsilon y$ with $y=\mathcal{O}(1)$ we obtain \eqref{eq:quasi-eq-sol-v} together with \eqref{eq:greens-fast-series} (with $s=s_2$)
\begin{equation}
v\sim \varepsilon^{-1}\omega_{s_1}\biggl(\sum_{j=1}^N\xi_j^2G_D(x_i,x_j) + \mathfrak{a}_{s_2} \xi_i^2 \varepsilon^{2s_2-1} |y|^{2s_2-1} + \varepsilon b_i y + \mathcal{O}(\varepsilon^{\min\{2,4s_2-1\}}) \biggr),
\end{equation}
where $b_i\equiv \sum_{j\neq i}\xi_j^2\nabla_1 G_D(x_i,x_j)$. It follows that the first order correction term in the inner expansion must be $\mathcal{O}(\varepsilon^{2s_2-1})$ and in particular for $x=x_i+\varepsilon y$ and $y=\mathcal{O}(1)$
\begin{equation*}
u\sim\varepsilon^{-1}\bigl(\xi_i w_{s_1}(y) + \varepsilon^{2s_2-1} U_{i1} + o(\varepsilon^{2s_2-1}) \bigr),\quad v\sim\varepsilon^{-1}\big(\xi_i + \varepsilon^{2s_2-1}V_{i1}  + o(\varepsilon^{2s_2-1})\bigr),
\end{equation*}
By repeatedly using the method of matched asymptotic expansions we determine that the fractional power $\varepsilon^{2s_2-1}$ initiates a chain of corrections at powers of $\varepsilon$ that are multiples of $2s_2-1$. In particular for each $i=1,...,N$ the inner expansion when $x=x_i+\varepsilon y$ with $y=\mathcal{O}(1)$ takes the form
\begin{gather}\label{eq:higher-order-inner}
u\sim \varepsilon^{-1}\bigl(\xi_iw_{s_1}(y) + \sum_{k=1}^{k_{\max}-1} \varepsilon^{k(2s_2-1)}U_{ik} + \varepsilon U_{ik_{\max}} + o(\varepsilon)  \bigr),\\ v\sim \varepsilon^{-1}\bigl(\xi_i + \sum_{k=1}^{k_{\max}-1} \varepsilon^{k(2s_2-1)}V_{ik} + \varepsilon V_{ik_{\max}} + o(\varepsilon)\bigr)
\end{gather}
where $k_{\max}$ is the smallest integer such that $k_{\max}(2s_2-1)\geq 1$. Importantly, since $V_{ik} \sim C_k|y|^{2s_2-1}$ as $|y|\rightarrow\infty$ for $1\leq k<k_{\max}$ each of these corrections are even in $y$. On the other hand when $k=k_{\max}$ we have the far-field behaviour
\begin{equation}
V_{ik_{\max}}\sim \omega_{s_1}b_i y  + \delta_{1,k_{\max}(2s_2-1)} C_{k_{\max}}|y|^{2s_2-1}\qquad |y|\rightarrow\infty,
\end{equation}
where $\delta_{i,j}$ is the discrete Kronecker delta function. Therefore we can write $V_{ik_{\max}}=\omega_{s_1}b_iy + V_{ik_{\max}}^e$ where $V_{ik_{\max}}^e$ is an even function in $y$. Assuming that each $x_i = x_i(t)$ and substituting \eqref{eq:higher-order-inner} into \eqref{1.fgm} with $x=x_i+\varepsilon y$ we obtain
\begin{equation}\label{eq:slow-temp-eq-1}
-\frac{1}{\varepsilon}\xi_i\frac{dw_{s_1}}{dy}\frac{dx_i}{dt}+ \sum_{k=1}^{k_{\max}-1}\varepsilon^{k(2s_2-1)}L_0 U_{ik} + \varepsilon L_0 U_{ik_{\max}} + \mathcal{N}_\varepsilon + \varepsilon w_{s_1}^2 \bigl( \omega_{s_1}b_iy + V_{ik_{\max}}^e \bigr) + o(\varepsilon) = 0,
\end{equation}
where $\mathcal{N}_\varepsilon$ is an even function of $y$ that consists of the residual nonlinear combinations of $U_{ik}$ and $V_{ik}$ for $1\leq k<k_{\max}$. Recalling that $dw_{s_1}/dy$ spans the kernel of $L_0$ we impose a solvability condition on \eqref{eq:slow-temp-eq-1} by multiplying it with $dw_{s_1}/dy$ and integrating to obtain
\begin{equation*}
\frac{dx_i}{dt} = \varepsilon^2 \frac{\omega_{s_1}b_i\int_{-\infty}^\infty w_{s_1}^2 \tfrac{dw_{s_1}}{dy} y dy}{\xi_i \int_{-\infty}^\infty |dw_{s_1}/dy|^2 dy} = -\varepsilon^2 \frac{\omega_{s_1}\int_{-\infty}^\infty w_{s_1}^3}{3\xi_i \int_{-\infty}^\infty|dw_{s_1}/dy|^2dy}b_i
\end{equation*}
where we have used integration by parts to obtain the second equality. This establishes \eqref{eq:slow-dynamics}.

\end{document}